\documentclass[notitlepage,aip,reprint,onecolumn]{revtex4-1}

\usepackage{amsmath,amssymb,amsfonts,amsthm,mathtools,bm}
\usepackage{bm}
\usepackage{array}
\DeclareMathSymbol{\hbar}{\mathord}{AMSb}{"7E}
\usepackage{mathptmx}
\usepackage{braket}
\usepackage{mathrsfs}
\usepackage{orcidlink}
\usepackage{lmodern}
\usepackage{microtype}
\usepackage{physics}
\usepackage{enumitem}
\usepackage{array}
\usepackage{longtable}
\usepackage{xcolor}
\usepackage{hyperref}

\newtheorem{theorem}{Teorema}[section]
\newtheorem{proposition}[theorem]{Proposition}
\newtheorem{lemma}[theorem]{Lemma}

\theoremstyle{definition}
\newtheorem{definition}[theorem]{Definition}
\newtheorem{remark}[theorem]{Remark}

\newcommand{\HH}{\mathcal H}
\newcommand{\BH}{\mathcal B(\HH)}

\newcommand{\HS}{\mathrm{HS}}
\renewcommand{\Tr}{\operatorname{Tr}}

\newcommand{\cD}{\mathcal D}

\newcommand{\braketHS}[2]{\left\langle #1,#2\right\rangle_{\HS}}
\renewcommand{\comm}[2]{\left[#1,#2\right]}
\newcommand{\acom}[2]{\left\{#1,#2\right\}}
\newcommand{\circJ}{\mathbin{\circ}}

\usepackage{graphicx}

\begin{document}


\title{A Lie-Jordan Geometric Formulation of Lindblad Dynamics}

\author{Leonel Bixano \orcidlink{0009-0009-5858-6466}}
\email{leonel.delacruz@cinvestav.mx}
\affiliation{Departamento de F\'{\i}sica,
Centro de Investigaci\'on y de Estudios Avanzados del
Instituto Politécnico Nacional,
Av. Instituto Politécnico Nacional 2508,
San Pedro Zacatenco,  CDMX, M\'exico 07360.}
\author{Victor Alberto Cruz-Barriguete\orcidlink{0000-0002-3321-924X}}
\author{Guillermo López-Alvarez\orcidlink{0009-0006-3646-1505}}
\author{V. G. Ibarra-Sierra\orcidlink{0000-0003-2416-3261}}
\author{José Luis Cardoso\orcidlink{0000-0003-3060-1168}}
\author{Juan Carlos Sandoval-Santana\orcidlink{0000-0002-1166-4529}}
\author{Alejandro Kunold\orcidlink{0000-0003-1545-8017}}
\email{akb@azc.uam.mx}
\affiliation{\'Area de F\'isica Te\'orica y Materia Condensada,
Universidad Aut\'onoma Metropolitana,
Azcapotzalco, Av. San Pablo Xalpa 180,
CDMX, 02200, M\'exico}


\begin{abstract}
We develop a Lie-Jordan geometric formulation of finite-dimensional open quantum dynamics, building on the algebraic framework introduced in our previous work (arXiv:2606.26477). The Hilbert-Schmidt operator space is endowed with an orthonormal Hermitian basis, in which the commutator and anticommutator are encoded by the structure tensors \(C_{\mu\nu}{}^\lambda\) and \(B_{\mu\nu}{}^\lambda\). Within this formulation, the von Neumann and Gorini-Kossakowski-Lindblad-Sudarshan equations admit a direct component representation.
Our central result is the identification of a basis-independent universal trilinear dissipative map,
\(
\mathcal D(X,Y)Z=XZY-\frac12\{YX,Z\},
\)
whose components define a universal operator-space tensor depending only on the Lie-Jordan structure tensors. The physical dissipator is obtained by contracting this tensor with the expansion coefficients of the Lindblad operators and the Kossakowski matrix, thereby separating the universal algebraic structure from the model-dependent physical information.
We further show that the combinations \((B+C)\) and \((B-C)\) generate internal left and right transports, allowing the elementary dissipative map to be expressed as a left-right bimodule action corrected by an ordered Jordan contribution. The universal map satisfies basis-independent trace and Hermitian-conjugation identities, from which trace preservation, Hermiticity preservation of the complete dissipator, and the reality of its component representation in a Hermitian Hilbert-Schmidt basis follow. We also derive its Hilbert-Schmidt adjoint and illustrate the formalism for a qubit with pure dephasing and amplitude-damping channels. This construction provides a tensorial and affine-geometric interpretation of the universal superoperator structure derived in arXiv:2606.26477 while keeping the algebraic and model-dependent sectors explicitly separated.
\end{abstract}

\maketitle

\section{Introduction}

The geometric formulation of quantum mechanics has emerged as a powerful framework for analyzing the structure of both closed and open quantum systems. For closed systems, the time evolution of the density operator is described by the von Neumann equation, whose form is dictated by unitary dynamics and whose generator is given by the commutator with the Hamiltonian \cite{vonNeumann1971}.

From a geometric perspective, this evolution can be viewed as a Hamiltonian flow on the manifold of quantum states: the eigenvalue spectrum of the density operator is conserved, the trajectory remains on the associated unitary orbit, and the commutator encodes the relevant Poisson-like structure. This standpoint is part of the contemporary geometric approach to quantum mechanics, in which states, observables, Poisson tensors, Riemannian tensors, and stratified state spaces are regarded as intrinsic geometric entities rather than purely algebraic constructs \cite{Ashtekar:1997ud,Grabowski:2005my,BengtssonZyczkowski2017}.

For open quantum systems, the scenario is more subtle. Assuming Markovianity, complete positivity, and trace preservation, the generator of the time evolution is given by the Gorini-Kossakowski-Lindblad-Sudarshan (GKLS) equation \cite{Kossakowski:1972sbn,Gorini:1975nb,Lindblad:1975ef}. In finite dimensions, this equation can be expressed as
\begin{equation}\label{eq:intro-GKLS}
    \dot\rho
    =
    -\frac{i}{\hbar}[H,\rho]
    +
    \sum_{\alpha,\beta}
    \gamma^{\alpha\beta}
    \left(
    L_\alpha\rho L_\beta^\dagger
    -\frac12\{L_\beta^\dagger L_\alpha,\rho\}
    \right),
\end{equation}
where \(H\) denotes the Hamiltonian, \(L_\alpha\) represent the Lindblad noise operators, and \(\gamma^{\alpha\beta}\) is the positive Kossakowski matrix. The first term corresponds to the unitary von Neumann evolution, whereas the second term captures dissipation, decoherence, and noise. The mathematical structure and physical relevance of this formulation are well established; see, for instance, contemporary reviews and pedagogical expositions on open quantum dynamics and the Lindblad master equation \cite{Breuer:2015zlm,Lidar:2019qog,Manzano:2020yyw}.

Previous work has shown that the manifold of quantum states carries a rich Lie-Jordan geometry generated by the commutator and anticommutator, leading to Hamiltonian and gradient flows on state space. The present work shifts the emphasis from the geometry of state space to the geometry of the Lindblad generator itself. We show that the GKLS dissipator is determined by a universal trilinear tensor defined solely by the Lie-Jordan algebra, with the physical Lindblad operators entering only through tensor contractions.

The aim of this work is thus neither to reestablish the GKLS theorem nor to assert novelty in the geometric viewpoint on open quantum dynamics. Instead, our focus is more narrowly defined: we construct an internal Lie–Jordan tensorial structure for the GKLS generator, tailored to an orthonormal operator basis, and employ it to define curvature-like objects linked to the dissipative sector.

The algebraic point of departure is the well-known fact that the self-adjoint part of an operator algebra naturally carries two different products. The antisymmetric product is the Lie product, given by the commutator, while the symmetric product is the Jordan product, given by the anticommutator.
This combined Lie-Jordan framework dates back to the work of
Jordan, von Neumann, and Wigner on the algebraic foundations
of quantum mechanics \cite{Jordan:1933vh}
and has been identified as an essential ingredient in the
description of the dynamics of quantum Markovian systems
\cite{KLendi_1987,alicki2007quantum}.
Its importance lies mainly in the fact that it allows
nonlinear matrix products appearing in the Lindblad equation
to be recast in algebraic form \cite{Byrd2003,Kimura2005}.
More recently, the combined use of the Lie and Jordan products
has also been employed to implement computationally efficient
expansions in a basis of $SU(N)$ generators  \cite{PhysRevE.100.053305}.
In an orthonormal Hilbert–Schmidt basis \(\{h_\mu\}\), these products can be represented via structure tensors as \([h_\mu,h_\nu] = C_{\mu\nu}{}^\lambda h_\lambda,\quad \{h_\mu,h_\nu\} = B_{\mu\nu}{}^\lambda h_\lambda,\) where \(C_{\mu\nu}{}^\lambda\) encodes the Lie structure and \(B_{\mu\nu}{}^\lambda\) encodes the Jordan structure. This simple observation will serve as the algebraic backbone of the construction developed in this paper.

The application of Lie–Jordan structures in quantum mechanics is, naturally, well established. The tensorial geometric description of quantum states and observables has been extensively developed within the geometric formulation of quantum mechanics and of density matrices \cite{Ashtekar:1997ud,Grabowski:2005my,Carinena:2009emx}. 
Similarly, geometric treatments of open quantum dynamics have demonstrated that the GKLS generator admits a decomposition into Hamiltonian, gradient-like, and intrinsically dissipative components on the manifold of quantum states 
\cite{Ciaglia:2017dvs,CHRUSCINSKI2019221,CRUZPRADO2025170123}. These investigations constitute the immediate precursors of the present framework.

We formulate finite-dimensional Lindblad dynamics directly using the Lie-Jordan tensors of an operator basis. The von Neumann part corresponds to the Lie sector, while the dissipative GKLS term is encoded in a tensor built from the same algebraic data. This allows a systematic separation of the purely Hamiltonian, Jordan, mixed Lie-Jordan, and noise-induced contributions to the evolution. The tensorial structure also suggests an intrinsic geometry of operator space, rather than spacetime. In particular, the tensors from the Lie and Jordan products define connection-like and curvature-like objects for the non-unitary evolution. These are not ordinary spacetime Riemann curvature tensors, but algebraic-geometric structures on operator space that diagnose compatibility or obstruction between the Hamiltonian Lie sector and the dissipative Jordan sector.

This viewpoint complements existing approaches to curvature in quantum information and open quantum dynamics. For example, Jordan-algebraic structures can generate Riemannian geometries on both classical and quantum state spaces \cite{Ciaglia:2017dvs}, and quantum Markov semigroups have been analyzed using tools such as noncommutative Ricci curvature, entropy-based techniques, and curvature-dimension inequalities \cite{Carlen:2017hgj}.

Moreover, motivated by the effectiveness of frame techniques in differential geometry, we explore how the internal tensors obtained in our framework can be projected onto suitably adapted complex frames. This construction is only conceptually reminiscent of the Newman–Penrose formalism in general relativity \cite{Newman:1961qr,Bixano:2026ouq,Bixano:2026xum}, rather than being directly equivalent to it.

The present work builds directly on \cite{Bixano:2026qce}, where a model independent algebraic formulation of finite dimensional Lindblad  dynamics was introduced. In that construction, the Liouville superoperator is expressed in terms of a closed algebra of Hermitian operators, while the specific physical model is encoded entirely in the corresponding expansion coefficients. This separation provides a universal basis for the dissipative dynamics and substantially simplifies the construction of the Liouville
generator. Here, we develop the geometric and tensorial content underlying that algebraic framework. In particular, we formulate the dynamics in the Hilbert-Schmidt operator space, express the elementary dissipative map through the Lie-Jordan structure tensors, and interpret the combinations \(B+C\) and \(B-C\) as internal left and right transports. This reformulation clarifies the origin of the universal superoperator basis and provides a direct geometric representation of the GKLS dissipative kernel.
The main result of this work is that the dissipative part of the GKLS
equation can be separated into two distinct levels.
At the universal level, it is completely characterized by a trilinear
map whose components depend only on the Lie and Jordan structure tensors
of the operator basis.
The specific physical model enters only through the expansion coefficients
of the Lindblad operators and the Kossakowski matrix \cite{Kossakowski:1972sbn}.
This separation cleanly distinguishes the universal algebraic structure 
of open-system dynamics from the model-dependent physical information.

The paper is structured as follows. In section \ref{Sec:Hilber-Schmidt metric and orthonormal basis}, we introduce the Hilbert-Schmidt geometry and an orthonormal Hermitian basis for the operator space. The associated Lie-Jordan structure is then developed in section \ref{Sec:Lie-Jordan structure of the operators space}. In sections \ref{Sec:von Neumann equation} and \ref{Sec:The GKLS equation in the LieJordan frame}, we formulate the von Neumann and GKLS equations within this framework and define the universal dissipative tensor. Section \ref{Sec:InternalTransports} introduces its internal left-right transport representation, affine geometric interpretation, and precise relation to ref \cite{Bixano:2026qce}. The structural properties of the universal dissipative map, including trace preservation, Hermitian conjugation, the reality of its component representation, and its Hilbert-Schmidt adjoint, are established in section \ref{Sec:StructuralPropertiesDissipativeMap}. Finally, section \ref{Sec:QubitRealization} illustrates the formalism explicitly for a qubit through pure dephasing and amplitude damping channels.

\section{Hilber-Schmidt metric and orthonormal operator basis}\label{Sec:Hilber-Schmidt metric and orthonormal basis}

\subsection{Convention}\label{SubSec:Conventions}
\begin{definition}[Hilbert space finite dimensional]
A \emph{finite-dimensional Hilbert space} is a complex vector space \(\mathcal H\) of finite dimension \(d\), equipped with a positive-definite inner product. In this work we always take
\[
     \mathcal H\simeq \mathbb C^d.
\]
For a qubit, \(d=2\). For a qutrit, \(d=3\). The finiteness of the dimension allows all operators to be represented as finite complex matrices.
\end{definition}
\begin{definition}[Complex vector space]
We denote by
\[
     \mathcal B(\mathcal H):=\{X:\mathcal H\to\mathcal H\mid X \text{ is linear}\}
\]
the complex vector space of all linear operators on \(\mathcal H\). Since \(\dim_{\mathbb C}\mathcal H=d\), we have the identification
\[
     \mathcal B(\mathcal H)\simeq M_d(\mathbb C),
\]
where \(M_d(\mathbb C)\) is the space of all complex matrices of size \(d\times d\). Therefore,
\[
     \dim_{\mathbb C}\mathcal B(\mathcal H)=d^2.
\]
\end{definition}
\begin{definition}[Adjoint, trace and identity]
For \(X\in\mathcal B(\mathcal H)\), we denote its Hermitian adjoint by \(X^\dagger\). In matrix form, \(X^\dagger\) is the conjugate transpose of the matrix. The trace is denoted by \(\operatorname{Tr}(X)\), and the identity on \(\mathcal H\) by \(\mathbb I\).
\end{definition}

\begin{definition}[Hermitian operators and density matrices]
The real subspace of Hermitian operators is defined by
\[
     \operatorname{Herm}(\mathcal H):=\{X\in\mathcal B(\mathcal H)\mid X^\dagger=X\}.
\]
An operator \(\rho\in\operatorname{Herm}(\mathcal H)\) is termed a density matrix if it satisfies
\[
     \rho\geq 0,
     \qquad
     \operatorname{Tr}(\rho)=1.
\]
The condition \(\rho \geq 0\) means that all eigenvalues of \(\rho\) are non-negative, or, equivalently, that \(\langle \psi, \rho \psi \rangle \geq 0\) holds for every \(\psi \in \mathcal H\). The collection of such density matrices is denoted by
\[
     \mathcal S(\mathcal H):=\{\rho\in\operatorname{Herm}(\mathcal H)\mid \rho\geq 0,\ \operatorname{Tr}\rho=1\}.
\]
\end{definition}
\begin{definition}[Index convention]
We will use Greek indices
\[
     \mu,\nu,\lambda,\sigma,\delta=0,1,\ldots,d^2-1
\]
to label a complete basis of \(\mathcal B(\mathcal H)\). The index \(0\) is usually reserved for the normalized identity. We will use Latin indices
\[
     i,j,k,a,b,c=1,\ldots,d^2-1
\]
for the traceless part. Whenever an index appears once upstairs and once downstairs, the Einstein summation convention will be used.
\end{definition}
\subsection{Hilbert--Schmidt inner product}\label{SubSec:Inner product of Hilbert-Schmidt}
The Hilbert–Schmidt product on \(\BH\) is defined as follows
\begin{equation} \label{eq:HS-product}
     \braketHS{A}{B}:=\Tr(A^\dagger B),
\end{equation}
The Hilbert–Schmidt product makes \(\mathcal B(\mathcal H)\) a finite-dimensional complex Hilbert space, and its restriction to \(\operatorname{Herm}(\mathcal H)\) gives a real Euclidean inner product. This structure appears in quantum-state geometry, Bloch and qudit representations, and superoperator formalisms, and here supplies the metric for defining operator components, dual coframes, and the Lie–Jordan structure tensors \cite{Grabowski:2005my,BengtssonZyczkowski2017,Nambu:2005amn}. The product is conjugate linear in its first argument, linear in its second, and positive definite:
\[
    \langle A,A\rangle_{\mathrm{HS}}
    =
    \operatorname{Tr}(A^\dagger A)
    \geq 0,
\]
with equality if and only if \(A=0\).
\subsection{Orthonormal basis}\label{SubSec:Orthonormal basis}
Let $\{h_\mu\}_{\mu=0}^{d^2-1}$ be an orthonormal basis of $\BH$ with respect to \eqref{eq:HS-product}, i.e. \(\braketHS{h_\mu}{h_\nu}=\delta_{\mu\nu}\), when the basis is hermitian $h_\mu^\dagger=h_\mu$, then \( \braketHS{h_\mu}{h_\nu}=\Tr(h_\mu h_\nu)=\delta_{\mu\nu}\). Therefore, every operator $X\in\BH$ admits the unique expansion
\begin{equation} \label{eq:operator-expansion}
     X=X^\mu h_\mu,
\end{equation}
with coefficients \( X^\mu=\braketHS{h_\mu}{X}=\Tr(h_\mu^\dagger X)\). If $X$ and $h_\mu$ are Hermitians, these coefficients are real.

Let the coframe basis $\{\theta^\mu\}$ be defined by
\begin{equation} \label{eq:theta-definition}
     \theta^\mu(X):=\braketHS{h_\mu}{X}=\Tr(h_\mu^\dagger X).
\end{equation}
Thus \(\theta^\mu(h_\nu)=\delta^\mu{}_{\nu}\), therefore \( X=X^\mu h_\mu,\quad X^\mu=\theta^\mu(X)\).

The Hilbert-Schmidt metric in the basis $\{h_\mu\}$ has components \( g_{\mu\bar\nu}:=\braketHS{h_\mu}{h_\nu}=\delta_{\mu\nu}\), if $X=X^\mu h_\mu$ and $Y=Y^\nu h_\nu$, then
\begin{equation} \label{eq:HS-components}
    \braketHS{X}{Y}=\overline{X^\mu}Y^\nu g_{\mu\bar\nu}
     =\sum_{\mu=0}^{d^2-1}\overline{X^\mu}Y^\mu.
\end{equation}
For Hermitian operators, the coefficients are real and the metric reduces to the standard Euclidean metric in real coordinates.

From now on, we choose the Hermitian Hilbert--Schmidt basis such that
\[
    h_0=\frac{\mathbb I}{\sqrt d},
    \qquad
    \operatorname{Tr}(h_i)=0,
    \qquad
    i=1,\ldots,d^2-1.
\]
Accordingly, \(\operatorname{Tr}(X)=\sqrt d\,X^0\), so \(X^0\) is the identity component, and the other coefficients describe the traceless sector.

The operator space contains \(d^2\) independent basis directions. Since a linear superoperator acts on this \(d^2\)-dimensional space,
\[
    \dim_{\mathbb C}
    \operatorname{End}\!\bigl(\mathcal B(\mathcal H)\bigr)
    =
    d^4.
\]
Thus a universal basis for arbitrary linear superoperators has \(d^4\) elements. In the notation of Ref.~\cite{Bixano:2026qce}, the Hilbert space has dimension \(m\), the operator space has dimension \(n=m^2\), and the universal superoperator space has dimension \(n^2=m^4\).

Specific physical models may use fewer elements if their coefficients are sparse or if symmetries, selection rules, or invariant subspaces constrain the dynamics.

\section{Lie-Jordan structure of the operator space}\label{Sec:Lie-Jordan structure of the operators space}
The Lie–Jordan structure used in this section is well known. The symmetric product of observables belongs to the tradition initiated by Jordan, von Neumann, and Wigner \cite{Jordan:1933vh}, while its modern tensorial formulation on state and observable spaces appears in the works of Grabowski–Ku\'s–Marmo, Cariñena et al., and Ciaglia et al. \cite{Grabowski:2005my,Carinena:2009emx,Ciaglia:2017dvs}. The novelty here is not the existence of these products, but the systematic use of their structure constants $B$ and $C$ to construct the internal tensor of the dissipator \cite{Byrd2003,Kimura2005}.
Some of the algebraic identities derived below were already
obtained in Ref.~\cite{Bixano:2026qce}.
Here, however, they are revisited from a geometric perspective,
where they naturally acquire the interpretation of tensorial
objects associated with the dissipator.

The ordinary associative product of operators $XY$ has two parts $\comm{X}{Y}:=XY-YX$, and $ \acom{X}{Y}:=XY+YX$, and the jordan product is defined by \( X\circJ Y:=\frac12\acom{X}{Y}\). Therefore, in a basis $\{h_\mu\}$, the constants structure are
\begin{align}
     \comm{h_\mu}{h_\nu}&=C_{\mu\nu}{}^\lambda h_\lambda,
     \label{eq:C-definition}
     \\
     \acom{h_\mu}{h_\nu}&=B_{\mu\nu}{}^\lambda h_\lambda,
     \label{eq:B-definition}
\end{align}
where by definition \(C_{\mu\nu}{}^\lambda=-C_{\nu\mu}{}^\lambda\) and \(B_{\mu\nu}{}^\lambda=B_{\nu\mu}{}^\lambda\). Moreover, in the coframe they are given by \(C_{\mu\nu}{}^\lambda=\theta^\lambda\big(\comm{h_\mu}{h_\nu}\big)\) and \(B_{\mu\nu}{}^\lambda=\theta^\lambda\big(\acom{h_\mu}{h_\nu}\big)\).

The associative product can be reconstructed from the Lie and Jordan products. In particular,
\begin{subequations}\label{eq:ordered-products-BC}
\begin{align}
    h_\mu h_\nu
    &=
    \frac12
    \left(
    B_{\mu\nu}{}^\lambda
    +
    C_{\mu\nu}{}^\lambda
    \right)h_\lambda,
    \label{eq:hmu-hnu-BC}
    \\
    h_\nu h_\mu
    &=
    \frac12
    \left(
    B_{\mu\nu}{}^\lambda
    -
    C_{\mu\nu}{}^\lambda
    \right)h_\lambda.
    \label{eq:hnu-hmu-BC}
\end{align}
\end{subequations}
These two ordered-product relations are the algebraic source of the \(B+C\) and \(B-C\) combinations that subsequently emerge in the GKLS dissipative kernel.
\subsection{Super-operators}\label{SubSec:Super-operators}
Let the linear super-operators over $\BH$ 
\begin{align}
     e_\mu(X)&:=\comm{h_\mu}{X},
     \label{eq:e-definition}
     \\
     \widetilde e_\mu(X)&:=\acom{h_\mu}{X}.
 \label{eq:etilde-definition}
\end{align}
Then acting over the basis \(h_\nu\), we can see that \( e_\mu(h_\nu)=C_{\mu\nu}{}^\lambda h_\lambda\), and \(\widetilde e_\mu(h_\nu)=B_{\mu\nu}{}^\lambda h_\lambda\), thus the matrix representation of this objects is
\begin{align}
     (e_\mu)^\lambda{}_{\nu}&=C_{\mu\nu}{}^\lambda,
     \\
     (\widetilde e_\mu)^\lambda{}_{\nu}&=B_{\mu\nu}{}^\lambda.
\end{align}
For an arbitrary operator \(A=A^\mu h_\mu\in\mathcal B(\mathcal H)\), thus we define
\begin{equation}\label{eq:eA-etildeA-definition}
    e_A
    :=
    A^\mu e_\mu,
    \qquad
    \widetilde e_A
    :=
    A^\mu\widetilde e_\mu.
\end{equation}
Their action on \(X\in\mathcal B(\mathcal H)\) is
\begin{equation}\label{eq:eA-etildeA-action}
    e_A(X)
    =
    [A,X],
    \qquad
    \widetilde e_A(X)
    =
    \{A,X\}.
\end{equation}
By linearity, the superoperators \(e_\mu\) and \(\widetilde e_\mu\) generate the Lie and Jordan actions of any operator in \(\mathcal B(\mathcal H)\). The maps \(e_\mu\) encode the adjoint (Lie) action of the operator basis and generate the Hamiltonian part of the von Neumann equation. In contrast, \(\widetilde e_\mu\) describe Jordan multiplication by the basis elements, which contributes to the dissipative part of the GKLS generator. However, a single \(\widetilde e_\mu\) is not itself a valid dissipative generator, as it need not preserve trace, positivity, or complete positivity. The physical GKLS dissipator arises from a specific combination of Lie and Jordan actions.

Now, let $X\in\BH$, then \( \comm{e_\mu}{e_\nu}(X) =e_\mu(e_\nu(X))-e_\nu(e_\mu(X)) =\comm{h_\mu}{\comm{h_\nu}{X}} -\comm{h_\nu}{\comm{h_\mu}{X}}\), using the Jacobi identity 
\[
     \comm{h_\mu}{\comm{h_\nu}{X}}
     -\comm{h_\nu}{\comm{h_\mu}{X}}
     =\comm{\comm{h_\mu}{h_\nu}}{X},
\]
and by applying \eqref{eq:C-definition}, we obtain
\begin{align*}
     \comm{e_\mu}{e_\nu}(X)=\comm{\comm{h_\mu}{h_\nu}}{X}=\comm{C_{\mu\nu}{}^\lambda h_\lambda}{X}=C_{\mu\nu}{}^\lambda\comm{h_\lambda}{X}=C_{\mu\nu}{}^\lambda e_\lambda(X).
\end{align*}
Therefore, the identity holds \cite{georgi2000lie}
\begin{equation} \label{eq:e-e-closure}
 \comm{e_\mu}{e_\nu}=C_{\mu\nu}{}^\lambda e_\lambda.
\end{equation}
\begin{lemma}[Lie-Jordan compatibility]
For any $A,B,C\in\BH$,
\begin{equation}
     \comm{A}{\acom{B}{C}} \label{eq:comm-derivation-anticomm}
     =\acom{\comm{A}{B}}{C}+\acom{B}{\comm{A}{C}}.
\end{equation}
\end{lemma}
\begin{proof}
Expanding the left hand side,
\begin{align*}
     \comm{A}{\acom{B}{C}}
     &=A(BC+CB)-(BC+CB)A
     \\
     &=ABC+ACB-BCA-CBA.
\end{align*}
On the other hand 
\begin{align*}
     \acom{\comm{A}{B}}{C}+\acom{B}{\comm{A}{C}}
     &=(AB-BA)C+C(AB-BA)
     \\
     &\quad +B(AC-CA)+(AC-CA)B
     \\
     &=ABC-BAC+CAB-CBA
     \\
     &\quad +BAC-BCA+ACB-CAB
     \\
     &=ABC+ACB-BCA-CBA.
\end{align*}
Therefore the both sides coincide.
\end{proof}
\begin{proposition}[Mixture closure]
The super-operators $\widetilde e_\mu$ and $e_\nu$ fulfill
\begin{equation} \label{eq:etilde-e-closure}
 \comm{\widetilde e_\mu}{e_\nu}=C_{\mu\nu}{}^\lambda\widetilde e_\lambda.
\end{equation}
Or equivalently, using $\comm{A}{B}=-\comm{B}{A}$ 
\[
     \comm{e_\nu}{\widetilde e_\mu}=C_{\nu\mu}{}^\lambda\widetilde e_\lambda.
\]
\end{proposition}

\begin{proof}
Taking $X\in\BH$, then
\begin{align*}
     \comm{\widetilde e_\mu}{e_\nu}(X)
     &=\widetilde e_\mu(e_\nu(X))-e_\nu(\widetilde e_\mu(X))
     \\
     &=\acom{h_\mu}{\comm{h_\nu}{X}}
     -\comm{h_\nu}{\acom{h_\mu}{X}}
     \\
     &=\acom{h_\mu}{\comm{h_\nu}{X}}
     -\acom{h_\mu}{\comm{h_\nu}{X}}
     -\acom{\comm{h_\nu}{h_\mu}}{X}
     \\
     &=-\acom{\comm{h_\nu}{h_\mu}}{X}=\acom{\comm{h_\mu}{h_\nu}}{X}=\acom{C_{\mu\nu}{}^\lambda h_\lambda}{X}
     \\
     &=C_{\mu\nu}{}^\lambda\widetilde e_\lambda(X).
\end{align*}
Where we have used \eqref{eq:comm-derivation-anticomm} with $A=h_\nu$, $B=h_\mu$, $C=X$, then \( \comm{h_\nu}{\acom{h_\mu}{X}}=\acom{\comm{h_\nu}{h_\mu}}{X}+\acom{h_\mu}{\comm{h_\nu}{X}}\).
Then \eqref{eq:etilde-e-closure} is proven.
\end{proof}
\begin{proposition}[Jordan closure]
The Jordan super-operators fulfill
\begin{equation} \label{eq:etilde-etilde-closure}
     \comm{\widetilde e_\mu}{\widetilde e_\nu}=C_{\mu\nu}{}^\lambda e_\lambda.
\end{equation}
\end{proposition}
\begin{proof}
Let $X\in\BH$,
\begin{align*}
     \widetilde e_\mu(\widetilde e_\nu(X))
     &=\acom{h_\mu}{\acom{h_\nu}{X}}
     \\
     &=h_\mu(h_\nu X+Xh_\nu)+(h_\nu X+Xh_\nu)h_\mu
     \\
     &=h_\mu h_\nu X+h_\mu Xh_\nu+h_\nu Xh_\mu+Xh_\nu h_\mu.
\end{align*}
On the other hand,
\begin{align*}
     \widetilde e_\nu(\widetilde e_\mu(X))
     &=h_\nu h_\mu X+h_\nu Xh_\mu+h_\mu Xh_\nu+Xh_\mu h_\nu.
\end{align*}
Subtracting,
\begin{align*}
     \comm{\widetilde e_\mu}{\widetilde e_\nu}(X)
     &=(h_\mu h_\nu-h_\nu h_\mu)X+X(h_\nu h_\mu-h_\mu h_\nu)
     \\
     &=\comm{h_\mu}{h_\nu}X-X\comm{h_\mu}{h_\nu}
     \\
     &=\comm{\comm{h_\mu}{h_\nu}}{X}=C_{\mu\nu}{}^\lambda\comm{h_\lambda}{X}
     =C_{\mu\nu}{}^\lambda e_\lambda(X).
\end{align*}
The equality holds for all $X$.
\end{proof}
\begin{proposition}[Mixture tensor for Lie-Jordan]
Defining
\begin{equation} \label{eq:D-tensor}
     D_{\mu\nu\sigma}{}^\delta
     :=C_{\nu\sigma}{}^\kappa B_{\mu\kappa}{}^\delta
     +B_{\mu\sigma}{}^\kappa C_{\nu\kappa}{}^\delta.
\end{equation}
Then, the super-operators \(\widetilde e_\mu\) and \(e_\nu\), fulfill
\begin{equation}\label{eq:etilde-etilde-closure-conmutador}
     \acom{\widetilde e_\mu}{e_\nu}(h_\sigma)
     =D_{\mu\nu\sigma}{}^\delta h_\delta.
\end{equation}
\end{proposition}
\begin{proof}
Acting on $h_\sigma$,
\begin{align*}
     \acom{\widetilde e_\mu}{e_\nu}(h_\sigma)
     &=\widetilde e_\mu(e_\nu(h_\sigma))+e_\nu(\widetilde e_\mu(h_\sigma))
     \\
     &=\widetilde e_\mu(C_{\nu\sigma}{}^\kappa h_\kappa)
     +e_\nu(B_{\mu\sigma}{}^\kappa h_\kappa)
     \\
     &=C_{\nu\sigma}{}^\kappa B_{\mu\kappa}{}^\delta h_\delta
     +B_{\mu\sigma}{}^\kappa C_{\nu\kappa}{}^\delta h_\delta.
\end{align*}
comparing with \eqref{eq:etilde-etilde-closure-conmutador}, we obtain \eqref{eq:D-tensor}.
\end{proof}
The tensor \(D_{\mu\nu\sigma}{}^\delta\) is the component representation of the symmetrized composition between a Jordan multiplication and a Lie derivation:
\[
    \{\widetilde e_\mu,e_\nu\}(h_\sigma)
    =
    D_{\mu\nu\sigma}{}^\delta h_\delta.
\]
This should not be viewed as an obstacle to achieving algebraic closure. In fact, the commutator algebra generated by \(e_\mu\) and \(\widetilde e_\mu\) is already closed, as demonstrated by Eqs.~\eqref{eq:e-e-closure}, \eqref{eq:etilde-e-closure}, and \eqref{eq:etilde-etilde-closure}. Rather, the tensor \(D_{\mu\nu\sigma}{}^\delta\) encodes the mixed, symmetrized combinations that arise when the ordered operator products appearing in the GKLS dissipator are decomposed into their Lie and Jordan components.
This tensor will be used explicitly in Appendix \ref{App:SuperoperatorBlocks}, where it represents the mixed Jordan-Lie blocks in the superoperator decomposition of \(L_\alpha\rho L_\beta^\dagger\).
\section{Von Neumann equation}\label{Sec:von Neumann equation}
The Von Neumann equations is given by
\begin{equation}
     \dot\rho=-\frac{\mathrm{i}}{\hbar}\comm{H}{\rho}.
\end{equation}
Expanding \(H=H^\mu h_\mu, \quad \rho=\rho^\nu h_\nu\), we obtain
\begin{align*}
     \dot\rho
     &=-\frac{\mathrm{i}}{\hbar}H^\mu\rho^\nu\comm{h_\mu}{h_\nu}=-\frac{\mathrm{i}}{\hbar}H^\mu\rho^\nu C_{\mu\nu}{}^\lambda h_\lambda.
\end{align*}
Therefore, te vectorization is
\begin{equation} \label{eq:von-neumann-components}
     \dot\rho^\lambda=-\frac{\mathrm{i}}{\hbar}H^\mu C_{\mu\nu}{}^\lambda\rho^\nu.
\end{equation}

\paragraph{Trace preservation.}
Since \(h_0=\mathbb I/\sqrt d\), the trace of the density operator is determined by its identity component, \(\operatorname{Tr}(\rho)=\sqrt d\,\rho^0\). Furthermore,
\[
    C_{\mu\nu}{}^0
    =
    \operatorname{Tr}\!\left(
    h_0[h_\mu,h_\nu]
    \right)
    =
    \frac1{\sqrt d}
    \operatorname{Tr}[h_\mu,h_\nu]
    =
    0,
\]
where cyclicity of the trace has been used. Therefore, the identity component of the von Neumann equation satisfies \( \dot\rho^0=0\), which is precisely the component form of trace preservation.

\begin{remark}[A Lie-Jordan linear model]
A general linear combination of Lie and Jordan multiplications may be written as
\[
    \dot\rho
    =
    -\frac{\mathrm{i}}{\hbar}[H,\rho]
    +
    \frac1{\hbar}\{G,\rho\}.
\]
In the Hilbert-Schmidt basis, this equation becomes
\[
    \dot\rho^\lambda
    =
    -\frac{\mathrm{i}}{\hbar}
    H^\mu C_{\mu\nu}{}^\lambda\rho^\nu
    +
    \frac1{\hbar}
    G^\mu B_{\mu\nu}{}^\lambda\rho^\nu.
\]
This equation is an intermediate geometric model included only to illustrate separately the actions of the Lie and Jordan tensors. In general, it is not a GKLS equation and does not by itself ensure trace preservation, positivity, or complete positivity.
\end{remark}
\section{The GKLS equation in the Lie-Jordan frame}\label{Sec:The GKLS equation in the LieJordan frame}
The finite-dimensional GKLS equation can be written as
\begin{equation}\label{eq:GKLS}
    \dot\rho
    =-\frac{\mathrm{i}}{\hbar}\comm{H}{\rho}
    +\gamma^{\alpha\beta}\mathfrak{D}_{\alpha\beta}(\rho),
\end{equation}
where
\begin{equation}\label{eq:D-alpha-beta}
    \mathfrak{D}_{\alpha\beta}(\rho)
    :=L_\alpha\rho L_\beta^\dagger
    -\frac12\acom{L_\beta^\dagger L_\alpha}{\rho}.
\end{equation}
Here, \(\gamma^{\alpha\beta}\) denotes a Hermitian, positive semidefinite Kossakowski matrix, and the operators \(L_\alpha\) are referred to as Lindblad operators, which, when the Kossakowski matrix is diagonal, may also be interpreted as quantum jump operators. The Hamiltonian part has already been expressed in \eqref{eq:von-neumann-components}. The remainder of this section derives the component representation of \(\cD_{\alpha\beta}\) using only the tensors \(B\) and \(C\).

To vectorize \eqref{eq:GKLS}, we start by expanding
\(
    H=H^\mu h_\mu,
    \quad
    L_\alpha=\ell_\alpha{}^\mu h_\mu,
    \quad
    L_\beta^\dagger=\overline{\ell_\beta{}^\nu}\,h_\nu,
    \quad
    \rho=\rho^\sigma h_\sigma
\), the conjugation appears because the basis \(h_\nu\) is Hermitian. 
The ordered-product identities \eqref{eq:ordered-products-BC} allow the first dissipative term to be evaluated directly. We have
\[
    L_\alpha\rho L_\beta^\dagger
    =
    \ell_\alpha{}^\mu
    \overline{\ell_\beta{}^\nu}
    \rho^\sigma
    h_\mu h_\sigma h_\nu.
\]
First,
\[
    h_\mu h_\sigma
    =
    \frac12
    \left(
    B_{\mu\sigma}{}^p
    +
    C_{\mu\sigma}{}^p
    \right)h_p.
\]
Multiplying the result from the right by \(h_\nu\), and using
\[
    h_p h_\nu
    =
    \frac12
    \left(
    B_{\nu p}{}^\delta
    -
    C_{\nu p}{}^\delta
    \right)h_\delta,
\]
we obtain
\begin{equation}\label{eq:LrhoL-BC-direct-expanded}
    L_\alpha\rho L_\beta^\dagger
    =
    \frac14
    \ell_\alpha{}^\mu
    \overline{\ell_\beta{}^\nu}
    \rho^\sigma
    \left(
    B_{\mu\sigma}{}^p
    +
    C_{\mu\sigma}{}^p
    \right)
    \left(
    B_{\nu p}{}^\delta
    -
    C_{\nu p}{}^\delta
    \right)
    h_\delta.
\end{equation}
The ordered product in the anticommutator is
\[
    L_\beta^\dagger L_\alpha
    =
    \ell_\alpha{}^\mu
    \overline{\ell_\beta{}^\nu}
    h_\nu h_\mu
    =
    \frac12
    \ell_\alpha{}^\mu
    \overline{\ell_\beta{}^\nu}
    \left(B_{\mu\nu}{}^m-C_{\mu\nu}{}^m\right)h_m.
\]
Therefore,
\begin{equation}\label{eq:anticommutator-short}
    \frac12\acom{L_\beta^\dagger L_\alpha}{\rho}
    =
    \frac14
    \ell_\alpha{}^\mu
    \overline{\ell_\beta{}^\nu}
    \rho^\sigma
    \left(B_{\mu\nu}{}^m-C_{\mu\nu}{}^m\right)
    B_{m\sigma}{}^\delta h_\delta.
\end{equation}

Equation~\eqref{eq:LrhoL-BC-direct-expanded} follows directly from the ordered product identities \(B\pm C\). An independent derivation based on the complete decomposition into Lie and Jordan superoperator blocks is presented in Appendix \ref{App:SuperoperatorBlocks}. There, the closure
relations of Sec \ref{SubSec:Super-operators} are used to recover exactly the same factorized expression.
%

\subsection{Universal elementary dissipative map and physical kernel}

The previous calculation reveals that every elementary
dissipative contribution appearing in the GKLS equation
possesses the same algebraic structure,
independently of the particular form of the Lindblad operators.
This observation motivates the introduction of a universal
trilinear dissipative map
\begin{equation}\label{eq:universal-D-map}
    \mathcal D(X,Y)Z
    :=
    XZY
    -
    \frac12\acom{YX}{Z},
\end{equation}
for three independent operators \(X,Y,Z\in\mathcal B(\mathcal H)\).
Equivalently, one may regard
\[
    \mathcal D(X,Y,Z)
    :=
    \mathcal D(X,Y)Z
\]
as a complex-trilinear map
\[
    \mathcal D:
    \mathcal B(\mathcal H)
    \times
    \mathcal B(\mathcal H)
    \times
    \mathcal B(\mathcal H)
    \longrightarrow
    \mathcal B(\mathcal H).
\]
For fixed \(X\) and \(Y\), the map \(Z\mapsto\mathcal D(X,Y)Z\) is a linear endomorphism of the operator space. The trilinearity refers to the independent arguments \(X,Y,Z\),
after the physical identification \(Y=L_\beta^\dagger\), the dependence on the coefficients of \(L_\beta\) is conjugate linear.

The physical dissipative block introduced in \eqref{eq:D-alpha-beta} is obtained by setting
\[
    X=L_\alpha,
    \qquad
    Y=L_\beta^\dagger,
\]
namely,
\begin{equation}\label{eq:physical-D-operator}
    \mathfrak{D}_{\alpha\beta}(Z)
    =
    \mathcal D(L_\alpha,L_\beta^\dagger)Z.
\end{equation}

For basis elements, we define the universal elementary maps
\begin{equation}\label{eq:elementary-D-munu}
    \mathcal D_{\mu\nu}(Z)
    :=
    \mathcal D(h_\mu,h_\nu)Z
    =
    h_\mu Z h_\nu
    -
    \frac12\acom{h_\nu h_\mu}{Z}.
\end{equation}
Using
\[
    L_\alpha
    =
    \ell_\alpha{}^\mu h_\mu,
    \qquad
    L_\beta^\dagger
    =
    \overline{\ell_\beta{}^\nu}\,h_\nu,
\]
the physical block becomes
\begin{equation}\label{eq:physical-D-from-universal}
    \mathfrak{D}_{\alpha\beta}(Z)
    =
    \ell_\alpha{}^\mu
    \overline{\ell_\beta{}^\nu}
    \mathcal D_{\mu\nu}(Z).
\end{equation}

The components of the universal map are defined by
\begin{equation}\label{eq:universal-D-components-def}
    \mathcal D_{\mu\nu}(h_\sigma)
    =
    \mathcal D_{\mu\nu\sigma}{}^\delta h_\delta.
\end{equation}
Combining \eqref{eq:LrhoL-BC-direct-expanded} and
\eqref{eq:anticommutator-short}, one obtains
\begin{equation}\label{eq:universal-D-components}
    \mathcal D_{\mu\nu\sigma}{}^\delta
    =
    \frac14
    \left[
    \left(
    B_{\mu\sigma}{}^p
    +
    C_{\mu\sigma}{}^p
    \right)
    \left(
    B_{\nu p}{}^\delta
    -
    C_{\nu p}{}^\delta
    \right)
    -
    \left(
    B_{\mu\nu}{}^m
    -
    C_{\mu\nu}{}^m
    \right)
    B_{m\sigma}{}^\delta
    \right].
\end{equation}

The associated kernel, resolved in terms of the Lindblad operator, is
\begin{equation}\label{eq:frakD-final}
    \mathfrak{D}_{\alpha\beta\,\sigma}{}^\delta
    =
    \ell_\alpha{}^\mu
    \overline{\ell_\beta{}^\nu}
    \mathcal D_{\mu\nu\sigma}{}^\delta.
\end{equation}
Consequently,
\[
    \mathfrak{D}_{\alpha\beta}(\rho)
    =
    \rho^\sigma
    \mathfrak{D}_{\alpha\beta\,\sigma}{}^\delta h_\delta.
\]

The full component representation of the GKLS equation is therefore
\begin{equation}\label{eq:full-GKLS-components}
    \dot\rho^\delta
    =
    -\frac{\mathrm{i}}{\hbar}
    H^\mu C_{\mu\sigma}{}^\delta\rho^\sigma
    +
    \gamma^{\alpha\beta}
    \mathfrak{D}_{\alpha\beta\,\sigma}{}^\delta\rho^\sigma.
\end{equation}

Equations from \eqref{eq:universal-D-map} to \eqref{eq:full-GKLS-components} disentangle the various structural layers of the dissipative dynamics. The map \(\mathcal D(X,Y,Z)\) is universal and independent of any specific basis, \(\mathcal D_{\mu\nu\sigma}{}^\delta\) encodes its components in operator space, the coefficients \(\ell_\alpha{}^\mu\) fix the Lindblad operators, and the Kossakowski matrix \(\gamma^{\alpha\beta}\) carries out the final physically meaningful contraction.

\begin{proposition}[Universal trace constraint of the dissipative map]
For arbitrary operators \(X,Y,Z\in\mathcal B(\mathcal H)\), the universal dissipative map satisfies
\begin{equation}\label{eq:universal-trace-constraint}
    \operatorname{Tr}\!\left[
    \mathcal D(X,Y)Z
    \right]
    =
    0.
\end{equation}
Therefore, for every fixed pair \(X,Y\), the image of \(\mathcal D(X,Y)\) is contained in the traceless operator subspace.
\end{proposition}

\begin{proof}
From the definition of the universal dissipative map,
\[
    \mathcal D(X,Y)Z
    =
    XZY
    -
    \frac12\acom{YX}{Z},
\]
we obtain
\begin{align*}
    \operatorname{Tr}\!\left[
    \mathcal D(X,Y)Z
    \right]
    &=
    \operatorname{Tr}(XZY)
    -
    \frac12\operatorname{Tr}(YXZ)
    -
    \frac12\operatorname{Tr}(ZYX).
\end{align*}
By cyclicity of the trace,
\[
    \operatorname{Tr}(XZY)
    =
    \operatorname{Tr}(YXZ),
    \qquad
    \operatorname{Tr}(ZYX)
    =
    \operatorname{Tr}(YXZ).
\]
Consequently,
\begin{align*}
    \operatorname{Tr}\!\left[
    \mathcal D(X,Y)Z
    \right]
    &=
    \operatorname{Tr}(YXZ)
    -
    \frac12\operatorname{Tr}(YXZ)
    -
    \frac12\operatorname{Tr}(YXZ)
    \\
    &=0.
\end{align*}
\end{proof}

\paragraph{Component representation.}
Once a Hilbert-Schmidt orthonormal basis is chosen, with
\[
    h_0=\frac{\mathbb I}{\sqrt d},
    \qquad
    \operatorname{Tr}(h_0)=\sqrt d,
    \qquad
    \operatorname{Tr}(h_a)=0
    \quad
    (a=1,\ldots,d^2-1),
\]
the universal elementary maps satisfy
\(
    \mathcal D_{\mu\nu}(h_\sigma)
    =
    \mathcal D_{\mu\nu\sigma}{}^\delta h_\delta.
\)
Taking the trace and using \eqref{eq:universal-trace-constraint}, one finds
\[
    0
    =
    \operatorname{Tr}
    \left[
    \mathcal D_{\mu\nu}(h_\sigma)
    \right]
    =
    \sqrt d\,
    \mathcal D_{\mu\nu\sigma}{}^0.
\]
Hence, \(\mathcal D_{\mu\nu\sigma}{}^0 = 0\), now contracting with the Lindblad-operator coefficients gives
\[
    \mathfrak{D}_{\alpha\beta\,\sigma}{}^0
    =
    \ell_\alpha{}^\mu
    \overline{\ell_\beta{}^\nu}
    \mathcal D_{\mu\nu\sigma}{}^0
    =
    0.
\]
Equivalently, since
\(
    \mathfrak{D}_{\alpha\beta}(Z)
    =
    \mathcal D(L_\alpha,L_\beta^\dagger)Z,
\)
the physical elementary blocks satisfy
\begin{equation*}
    \operatorname{Tr}\!\left[
    \mathfrak{D}_{\alpha\beta}(Z)
    \right]
    =
    0
\end{equation*}
for every \(Z\in\mathcal B(\mathcal H)\), and in particular for \(Z=\rho\).

\begin{remark}
An individual cross block \(\mathfrak{D}_{\alpha\beta}\) need not preserve Hermiticity when \(\alpha\neq\beta\). Hermiticity preservation is a property of the full contraction \(\gamma^{\alpha\beta}\mathfrak{D}_{\alpha\beta}\),  as a consequence of the Hermiticity of the Kossakowski matrix \(\gamma^{\alpha\beta}\).
\end{remark}

\section{Internal left-right Lie-Jordan transports}\label{Sec:InternalTransports}

The combinations \(B+C\) and \(B-C\) appearing in the universal dissipative tensor admit a direct algebraic interpretation: they represent, respectively, left and right multiplication in the associative operator algebra. To encode these two operations compactly, we introduce the internal transport operators
\begin{subequations}\label{eq:nabla-plus-minus}
\begin{align}
    \nabla^{(+)}_\mu h_\nu
    &:=
    \Gamma^{(+)}_{\mu\nu}{}^\sigma h_\sigma,
    &
    \Gamma^{(+)}_{\mu\nu}{}^\sigma
    &:=
    B_{\mu\nu}{}^\sigma
    +
    C_{\mu\nu}{}^\sigma,
    \label{eq:nabla-plus}
    \\
    \nabla^{(-)}_\mu h_\nu
    &:=
    \Gamma^{(-)}_{\mu\nu}{}^\sigma h_\sigma,
    &
    \Gamma^{(-)}_{\mu\nu}{}^\sigma
    &:=
    B_{\mu\nu}{}^\sigma
    -
    C_{\mu\nu}{}^\sigma.
    \label{eq:nabla-minus}
\end{align}
\end{subequations}
They are extended by linearity to every
\(Z=Z^\sigma h_\sigma\in\mathcal B(\mathcal H)\).

\begin{remark}[Algebraic and geometric interpretation]
The coefficients \(\Gamma^{(+)}_{\mu\nu}{}^\sigma\) and \(\Gamma^{(-)}_{\mu\nu}{}^\sigma\) encode two internal algebraic transports. As shown below, the corresponding operators coincide with twice the left and right regular actions, respectively. When \(\mathcal B(\mathcal H)\) is regarded as an affine manifold, these prescriptions admit a standard Leibniz extension to two linear  connections.
This geometric extension is made explicit below.
\end{remark}

\begin{proposition}[Left and right regular actions]
For every \(Z\in\mathcal B(\mathcal H)\),
\begin{subequations}\label{eq:left-right-actions}
\begin{align}
    \nabla^{(+)}_\mu Z
    &=
    2h_\mu Z,
    \label{eq:left-action}
    \\
    \nabla^{(-)}_\mu Z
    &=
    2Zh_\mu.
\label{eq:right-action}
\end{align}
\end{subequations}
Thus, \(\nabla^{(+)}_\mu\) and \(\nabla^{(-)}_\mu\) are twice the left and right regular actions of \(h_\mu\), respectively.
\end{proposition}

\begin{proof}
Using the definitions of the Lie and Jordan superoperators,
\begin{align*}
    \nabla^{(+)}_\mu Z
    &=
    \widetilde e_\mu(Z)+e_\mu(Z)
    \\
    &=
    \{h_\mu,Z\}+[h_\mu,Z]
    \\
    &=
    (h_\mu Z+Zh_\mu)+(h_\mu Z-Zh_\mu)
    \\
    &=
    2h_\mu Z.
\end{align*}
Similarly,
\begin{align*}
    \nabla^{(-)}_\mu Z
    &=
    \widetilde e_\mu(Z)-e_\mu(Z)
    \\
    &=
    \{h_\mu,Z\}-[h_\mu,Z]
    \\
    &=
    (h_\mu Z+Zh_\mu)-(h_\mu Z-Zh_\mu)
    \\
    &=
    2Zh_\mu.
\end{align*}
\end{proof}

Because \(B_{\mu\nu}{}^\sigma\) is symmetric and \(C_{\mu\nu}{}^\sigma\) is antisymmetric in their lower indices,
\[
\Gamma^{(-)}_{\mu\nu}{}^\sigma
=
\Gamma^{(+)}_{\nu\mu}{}^\sigma.
\]
The Jordan and Lie actions can therefore be recovered as
\begin{subequations}
\begin{align}
    \nabla^{(J)}_\mu
    &:=
    \frac12
    \left(
    \nabla^{(+)}_\mu+\nabla^{(-)}_\mu
    \right)
    =
    \widetilde e_\mu,
    \label{eq:nabla-J}
    \\
    \nabla^{(C)}_\mu
    &:=
    \frac12
    \left(
    \nabla^{(+)}_\mu-\nabla^{(-)}_\mu
    \right)
    =
    e_\mu.
    \label{eq:nabla-C}
\end{align}
\end{subequations}
Accordingly,
\[
    \nabla^{(J)}_\mu Z
    =
    \{h_\mu,Z\},
    \qquad
    \nabla^{(C)}_\mu Z
    =
    [h_\mu,Z].
\]

We also introduce the ordered Jordan correction
\begin{equation}\label{eq:anchor-D-def}
    \nabla^{(-J)}_{\mu\nu}
    :=
    \Gamma^{(-)}_{\mu\nu}{}^m
    \nabla^{(J)}_m.
\end{equation}
The term \emph{ordered} refers to the identity
\[
    \Gamma^{(-)}_{\mu\nu}{}^m h_m
    =
    2h_\nu h_\mu.
\]
Acting on an arbitrary operator \(Z\), one obtains
\begin{equation}\label{eq:ordered-Jordan-action}
    \nabla^{(-J)}_{\mu\nu}Z
    =
    \Gamma^{(-)}_{\mu\nu}{}^m
    \{h_m,Z\}
    =
    \left\{
    \Gamma^{(-)}_{\mu\nu}{}^m h_m,
    Z
    \right\}
    =2\{h_\nu h_\mu,Z\}.
\end{equation}

\subsection{Transport representation of the universal dissipative map}

The universal dissipative map was introduced in \eqref{eq:universal-D-map}. On the Hilbert-Schmidt basis, its elementary maps are
\[
    \mathcal D_{\mu\nu}(Z)
    =
    \mathcal D(h_\mu,h_\nu)Z
    =
    h_\mu Z h_\nu
    -
    \frac12\{h_\nu h_\mu,Z\}.
\]
The internal transport operators provide a direct representation of its two contributions. From  \eqref{eq:left-right-actions},
\[
    \frac14
    \nabla^{(-)}_\nu
    \nabla^{(+)}_\mu Z
    =
    h_\mu Z h_\nu,
\]
whereas  \eqref{eq:ordered-Jordan-action} gives
\[
    \frac14
    \nabla^{(-J)}_{\mu\nu}Z
    =
    \frac12\{h_\nu h_\mu,Z\}.
\]
Consequently,
\begin{equation}\label{eq:internal-D-operator-def}
    \mathcal D_{\mu\nu}(Z)
    =
    \frac14
    \left(
    \nabla^{(-)}_\nu\nabla^{(+)}_\mu
    -
    \nabla^{(-J)}_{\mu\nu}
    \right)Z.
\end{equation}

Equation \eqref{eq:internal-D-operator-def} provides the transport-form representation of the universal elementary dissipative map. In this formulation, the first term is interpreted as a left-right bimodule action, while the second term corresponds to the associated ordered Jordan correction. The operator-space components and their contraction with the Lindblad operator coefficients were derived in \eqref{eq:universal-D-components} and \eqref{eq:frakD-final}, respectively.

\subsection{Associative compatibility of the left and right transports}

\begin{proposition}[Mixed left-right compatibility]
The left and right transport operators commute:
\begin{equation}\label{eq:mixed-transport-commutator-zero}
    \left[
    \nabla^{(-)}_\nu,
    \nabla^{(+)}_\mu
    \right]
    =
    0
\end{equation}
for all \(\mu,\nu\).
\end{proposition}

\begin{proof}
For every \(Z\in\mathcal B(\mathcal H)\),
\begin{align*}
    \nabla^{(-)}_\nu\nabla^{(+)}_\mu Z
    &=
    4(h_\mu Z)h_\nu,
    \\
    \nabla^{(+)}_\mu\nabla^{(-)}_\nu Z
    &=
    4h_\mu(Zh_\nu).
\end{align*}
Associativity of the operator product implies
\[
    (h_\mu Z)h_\nu
    =
    h_\mu(Zh_\nu).
    \]
    Therefore,
    \[
    \nabla^{(-)}_\nu\nabla^{(+)}_\mu Z
    =
    \nabla^{(+)}_\mu\nabla^{(-)}_\nu Z
\]
for every \(Z\in\mathcal B(\mathcal H)\), which proves \eqref{eq:mixed-transport-commutator-zero}.
\end{proof}

In components, \eqref{eq:mixed-transport-commutator-zero} is equivalent to the associativity identity
\begin{equation}\label{eq:BC-associativity-identity}
\left(
B_{\mu\sigma}{}^p+C_{\mu\sigma}{}^p
\right)
\left(
B_{\nu p}{}^\delta-C_{\nu p}{}^\delta
\right)
=
\left(
B_{\nu\sigma}{}^p-C_{\nu\sigma}{}^p
\right)
\left(
B_{\mu p}{}^\delta+C_{\mu p}{}^\delta
\right).
\end{equation}

As a consequence of Eq.~\eqref{eq:BC-associativity-identity}, the transport representation of the elementary dissipator admits the equivalent symmetrized form
\begin{equation}\label{eq:D-symmetrized}
\mathcal D_{\mu\nu}(Z)
=
\frac18
\left(
\nabla^{(-)}_\nu\nabla^{(+)}_\mu
+
\nabla^{(+)}_\mu\nabla^{(-)}_\nu
-
2\nabla^{(-J)}_{\mu\nu}
\right)Z.
\end{equation}

\subsection{Tensorial character}

The complex trilinearity of the universal map \(\mathcal D(X,Y,Z)\), introduced in \eqref{eq:universal-D-map}, implies that its components \(\mathcal D_{\mu\nu\sigma}{}^\delta\) transform as those of a \((3,1)\)-tensor under changes of operator basis.

Let
\[
    h'_{\mu'}
    =
    R_{\mu'}{}^\mu h_\mu
\]
be another operator basis, with inverse transformation
\[
    h_\mu
    =
    (R^{-1})_\mu{}^{\mu'}h'_{\mu'}.
\]
If
\[
    \mathcal D(h_\mu,h_\nu)h_\sigma
    =
    \mathcal D_{\mu\nu\sigma}{}^\delta h_\delta,
\]
then the components in the transformed basis satisfy
\begin{equation}\label{eq:D-tensor-transformation}
    \mathcal D'_{\mu'\nu'\sigma'}{}^{\delta'}
    =
    R_{\mu'}{}^\mu
    R_{\nu'}{}^\nu
    R_{\sigma'}{}^\sigma
    (R^{-1})_\delta{}^{\delta'}
    \mathcal D_{\mu\nu\sigma}{}^\delta.
\end{equation}
If both bases are Hermitian and Hilbert-Schmidt orthonormal, then \(R\) is a real orthogonal matrix.

The physical kernel \eqref{eq:physical-D-from-universal} contains two distinct types of indices. The indices \(\mu,\nu,\sigma,\delta\) belong to the operator space, whereas \(\alpha,\beta\) label the Lindblad operators. Thus, the operator-space tensorial structure is carried by \(\mathcal D_{\mu\nu\sigma}{}^\delta\), while \(\mathfrak{D}_{\alpha\beta\,\sigma}{}^\delta\) results from its contraction
with the model dependent Lindblad-operator coefficients.

\subsection{Affine geometric interpretation}

The complex vector space \(\mathcal M := \mathcal B(\mathcal H)\) can be viewed as a complex affine manifold, where each tangent space is naturally identified with \(\mathcal B(\mathcal H)\), i.e., \(T_Z\mathcal M \simeq \mathcal B(\mathcal H)\) for all \(Z \in \mathcal M\). By associating each basis element \(h_\mu\) with the corresponding constant vector field, the Hilbert–Schmidt basis induces a global frame on \(\mathcal M\).

Given a vector field \(Y = Y^\nu h_\nu\), the algebraic rules in \eqref{eq:nabla-plus-minus} extend to a Leibniz-type action
\begin{equation}\label{eq:affine-connection-extension}
    \nabla^{(\pm)}_{h_\mu}Y
    :=
    h_\mu(Y^\sigma)h_\sigma
    +
    Y^\nu
    \Gamma^{(\pm)}_{\mu\nu}{}^\sigma h_\sigma,
\end{equation}
extended \(C^\infty\)-linearly in the lower argument. In this expression, \(h_\mu(Y^\sigma)\) represents the directional derivative of the coefficient function \(Y^\sigma\) in the direction of the constant vector field \(h_\mu\). The formula \eqref{eq:affine-connection-extension} thereby specifies two linear connections with coefficients 
\(\Gamma^{(+)}_{\mu\nu}{}^\sigma\) and \(\Gamma^{(-)}_{\mu\nu}{}^\sigma\) in the global Hilbert–Schmidt frame. When acting on constant vector fields, the derivative term disappears, and these connections reduce to the left and right regular actions given in Eq.~\eqref{eq:left-right-actions}.

For fixed \(X,Y\in\mathcal B(\mathcal H)\), the endomorphism
\[
    Z
    \longmapsto
    \mathcal D(X,Y)Z
\]
defines a constant \((1,1)\)-tensor field on the affine operator space. Equation \eqref{eq:internal-D-operator-def} expresses its elementary basis components through the two internal transports and the ordered Jordan correction.

The physical state space
\[
    \mathcal S(\mathcal H)
    =
    \left\{
    \rho\in\operatorname{Herm}(\mathcal H)
    \;\middle|\;
    \rho\geq0,\;
    \operatorname{Tr}\rho=1
    \right\}
\]
is a convex subset of the trace one Hermitian affine hyperplane. Its boundary is stratified by the rank of the density operator. The full GKLS vector field is tangent to the trace one Hermitian affine hyperplane and preserves the positive state space.

\subsection{Dictionary with the algebraic formulation of Kunold et al article}

Ref \cite{Bixano:2026qce} introduces the real Lie constants \(c_{\mu\nu\lambda}\), the Jordan constants \(b_{\mu\nu\lambda}\), and the superoperators
\[
    Z_\mu
    =
    B_\mu+\mathrm{i} C_\mu,
    \qquad
    X_{\mu\nu}
    =
    \frac14Z_\mu Z_\nu^*.
\]
Its convention
\[
    [h_\mu,h_\nu]
    =
    \mathrm{i} c_{\mu\nu}{}^\lambda h_\lambda
\]
is related to the convention used here by
\[
    C_{\mu\nu}{}^\lambda
    =
    \mathrm{i} c_{\mu\nu}{}^\lambda,
    \qquad
    B_{\mu\nu}{}^\lambda
    =
    b_{\mu\nu}{}^\lambda.
\]

\begin{table}[t]
    \centering
    \small
    \renewcommand{\arraystretch}{1.35}
\begin{tabular}{|c|c|c|}
    \hline
    \textbf{Present notation}
    &
    \textbf{Ref.}~\cite{Bixano:2026qce}
    &
    \textbf{Meaning}
    \\
    \hline
    \(B_{\mu\nu}{}^\lambda\)
    &
    \(b_{\mu\nu\lambda}\)
    &
    Jordan structure constants
    \\
    \hline
    \(C_{\mu\nu}{}^\lambda\)
    &
    \(\mathrm{i} c_{\mu\nu\lambda}\)
    &
    Lie structure constants in the Hermitian-basis convention
    \\
    \hline
    \(\nabla^{(+)}_\mu\)
    &
    \(Z_\mu^*\)
    &
    Twice the left action:
    \(Z\mapsto2h_\mu Z\)
    \\
    \hline
    \(\nabla^{(-)}_\nu\)
    &
    \(Z_\nu\)
    &
    Twice the right action:
    \(Z\mapsto2Zh_\nu\)
    \\
    \hline
    \(\frac14\nabla^{(-)}_\nu\nabla^{(+)}_\mu\)
    &
    \(X_{\nu,\mu}
    =\frac14Z_\nu Z_\mu^*\)
    &
    Bimodule action:
    \(Z\mapsto h_\mu Z h_\nu\)
    \\
    \hline
    \(\nabla^{(-J)}_{\mu\nu}\)
    &
    \(z_{\nu\mu p}B_p\)
    &
    Ordered Jordan correction
    \\
    \hline
    \(\mathcal D_{\mu\nu}\)
    &
    \(\frac14
    \left(
    Z_\nu Z_\mu^*
    -
    z_{\nu\mu p}B_p
    \right)\)
    &
    Universal elementary dissipative map
    \\
    \hline
    \(\gamma^{\alpha\beta}\mathfrak{D}_{\alpha\beta}\)
    &
    \(\mathcal{L}_L\)
    &
    Full dissipative Liouville generator
    \\
    \hline
    \end{tabular}
    \caption{Dictionary between the Lie-Jordan transport notation used here and the algebraic superoperator notation of Ref \cite{Bixano:2026qce}. The reversal \((\mu,\nu)\mapsto(\nu,\mu)\) in \(X_{\nu,\mu}\) follows from the ordered action \(h_\mu Z h_\nu\).}
    \label{tab:Kunold-Bixano-dictionary}
\end{table}

In particular, the commutation relation
\[
    [Z_\nu,Z_\mu^*]
    =
    0
\]
obtained in \cite{Bixano:2026qce} is exactly the same as the left-right compatibility condition
\[
    [\nabla^{(-)}_\nu,\nabla^{(+)}_\mu]
    =
    0
\]
established in \eqref{eq:mixed-transport-commutator-zero}. Consequently, the two approaches capture an identical universal dissipative structure: the original paper formulates it in terms of structure-constant supermatrices, while the present treatment employs Lie–Jordan tensors together with internal left-right transports.

\section{Structural properties of the universal dissipative map}\label{Sec:StructuralPropertiesDissipativeMap}

The universal trace constraint established in \eqref{eq:universal-trace-constraint} shows that the image of every elementary dissipative map lies in the traceless operator subspace. We now derive the complementary conjugation, Hermiticity, reality, and Hilbert-Schmidt duality properties of the same universal structure.

\subsection{Hermitian conjugation of the universal map}

\begin{proposition}[Hermitian conjugation]
For arbitrary operators \(X,Y,Z\in\mathcal B(\mathcal H)\), the universal dissipative map satisfies
\begin{equation}\label{eq:universal-Hermitian-conjugation}
    \left[
    \mathcal D(X,Y)Z
    \right]^\dagger
    =
    \mathcal D(Y^\dagger,X^\dagger)Z^\dagger.
\end{equation}
\end{proposition}

\begin{proof}
From the definition
\[
    \mathcal D(X,Y)Z
    =
    XZY
    -
    \frac12\{YX,Z\},
\]
we obtain
\begin{align*}
    \left[
    \mathcal D(X,Y)Z
    \right]^\dagger
    &=
    Y^\dagger Z^\dagger X^\dagger
    -
    \frac12
    \left\{
    (YX)^\dagger,Z^\dagger
    \right\}
    \\
    &=
    Y^\dagger Z^\dagger X^\dagger
    -
    \frac12
    \left\{
    X^\dagger Y^\dagger,Z^\dagger
    \right\}
    \\
    &=
    \mathcal D(Y^\dagger,X^\dagger)Z^\dagger.
\end{align*}
\end{proof}

For a Hermitian Hilbert-Schmidt basis, \(h_\mu^\dagger=h_\mu\), equation \eqref{eq:universal-Hermitian-conjugation} implies
\begin{equation}\label{eq:elementary-Hermitian-conjugation}
    \left[
    \mathcal D_{\mu\nu}(Z)
    \right]^\dagger
    =
    \mathcal D_{\nu\mu}(Z^\dagger).
\end{equation}
In particular, since \(\mathcal D_{\mu\nu}(h_\sigma) = \mathcal D_{\mu\nu\sigma}{}^\delta h_\delta\), we obtain the component identity
\begin{equation}\label{eq:D-component-conjugation}
    \overline{
    \mathcal D_{\mu\nu\sigma}{}^\delta
    }
    =
    \mathcal D_{\nu\mu\sigma}{}^\delta.
\end{equation}
Thus, exchanging the first two operator space indices is equivalent to complex conjugation of the universal dissipative tensor.

\subsection{Hermiticity preservation of the physical dissipator}

The physical elementary blocks are obtained from the universal map \eqref{eq:physical-D-operator}. Equation \eqref{eq:universal-Hermitian-conjugation} therefore gives
\begin{equation}\label{eq:physical-block-Hermitian-conjugation}
    \left[
    \mathfrak{D}_{\alpha\beta}(Z)
    \right]^\dagger
    =
    \mathfrak{D}_{\beta\alpha}(Z^\dagger).
\end{equation}

A single off-diagonal block \(\mathfrak{D}_{\alpha\beta}\) with \(\alpha \neq \beta\) does not, by itself, have to be Hermitian. Hermiticity is only restored after contracting with the Kossakowski matrix.

\begin{proposition}[Hermiticity preservation]
Let
\[
    \mathcal L_{\mathrm{diss}}(Z)
    :=
    \gamma^{\alpha\beta}
    \mathfrak{D}_{\alpha\beta}(Z),
\]
where \(\gamma^{\alpha\beta} = \overline{\gamma^{\beta\alpha}}\). Then
\begin{equation}\label{eq:full-dissipator-Hermitian-conjugation}
\left[
\mathcal L_{\mathrm{diss}}(Z)
\right]^\dagger
=
\mathcal L_{\mathrm{diss}}(Z^\dagger).
\end{equation}
Consequently, the physical dissipator maps Hermitian operators into Hermitian operators.
\end{proposition}

\begin{proof}
Using \eqref{eq:physical-block-Hermitian-conjugation},
\begin{align*}
    \left[
    \mathcal L_{\mathrm{diss}}(Z)
    \right]^\dagger
    &=
    \overline{\gamma^{\alpha\beta}}
    \left[
    \mathfrak{D}_{\alpha\beta}(Z)
    \right]^\dagger
    \\
    &=
    \overline{\gamma^{\alpha\beta}}
    \mathfrak{D}_{\beta\alpha}(Z^\dagger)
    \\
    &=
    \overline{\gamma^{\beta\alpha}}
    \mathfrak{D}_{\alpha\beta}(Z^\dagger)
    \\
    &=
    \gamma^{\alpha\beta}
    \mathfrak{D}_{\alpha\beta}(Z^\dagger)
    \\
    &=
    \mathcal L_{\mathrm{diss}}(Z^\dagger),
\end{align*}
where the indices \(\alpha\) and \(\beta\) were interchanged in the third line and the Hermiticity of the Kossakowski matrix was used in the fourth.
\end{proof}

\subsection{Reality of the component representation}

Define the dissipative component matrix by
\begin{equation}\label{eq:physical-dissipative-matrix}
    K_\sigma{}^\delta
    :=
    \gamma^{\alpha\beta}
    \mathfrak{D}_{\alpha\beta\,\sigma}{}^\delta.
\end{equation}
From \eqref{eq:D-component-conjugation} and the definition of the physical kernel \eqref{eq:physical-D-from-universal}, one finds
\(
\overline{
\mathfrak{D}_{\alpha\beta\,\sigma}{}^\delta
}
=
\mathfrak{D}_{\beta\alpha\,\sigma}{}^\delta.
\)
It follows that
\begin{align*}
\overline{K_\sigma{}^\delta}
&=
\overline{\gamma^{\alpha\beta}}\,
\overline{
\mathfrak{D}_{\alpha\beta\,\sigma}{}^\delta
}
\\
&=
\overline{\gamma^{\alpha\beta}}\,
\mathfrak{D}_{\beta\alpha\,\sigma}{}^\delta
\\
&=
\overline{\gamma^{\beta\alpha}}\,
\mathfrak{D}_{\alpha\beta\,\sigma}{}^\delta
\\
&=
\gamma^{\alpha\beta}
\mathfrak{D}_{\alpha\beta\,\sigma}{}^\delta
\\
&=
K_\sigma{}^\delta.
\end{align*}
Therefore,
\begin{equation}\label{eq:physical-dissipative-matrix-real}
K_\sigma{}^\delta
\in\mathbb R
\end{equation}
in every Hermitian Hilbert--Schmidt basis.

The Hamiltonian contribution is real in the same basis. Indeed, the coefficients \(H^\mu\) of a Hermitian Hamiltonian are real, whereas \(C_{\mu\sigma}{}^\delta\) is purely imaginary because it represents the commutator of Hermitian operators. Hence,
\[
    -\frac{\mathrm{i}}{\hbar}
    H^\mu C_{\mu\sigma}{}^\delta
    \in\mathbb R.
\]
Consequently, the complete GKLS component matrix
\begin{equation}\label{eq:full-real-GKLS-matrix}
    \mathbb L_\sigma{}^\delta
    :=
    -\frac{\mathrm{i}}{\hbar}
    H^\mu C_{\mu\sigma}{}^\delta
    +
    K_\sigma{}^\delta
\end{equation}
is real. Thus, although the density operator and the individual superoperator blocks may involve complex coefficients, the evolution of the components of a Hermitian density operator is governed by a real linear system.

\subsection{Hilbert--Schmidt adjoint}

The Hermitian conjugation relation \eqref{eq:universal-Hermitian-conjugation} must be distinguished from the adjoint of a superoperator with respect to the Hilbert-Schmidt inner product. For a linear map \(\mathcal S\), its Hilbert-Schmidt adjoint \(\mathcal S^\ddagger\) is defined by
\begin{equation}\label{eq:HS-superoperator-adjoint-def}
    \left\langle
    A,\mathcal S(B)
    \right\rangle_{\mathrm{HS}}
    =
    \left\langle
    \mathcal S^\ddagger(A),B
    \right\rangle_{\mathrm{HS}},
\end{equation}
where
\( \langle A,B\rangle_{\mathrm{HS}} = \operatorname{Tr}(A^\dagger B) \).

\begin{proposition}[Hilbert-Schmidt adjoint of the universal map]
For fixed \(X,Y\in\mathcal B(\mathcal H)\),
\begin{equation}\label{eq:HS-adjoint-universal-D}
    \mathcal D(X,Y)^\ddagger(A)
    =
    X^\dagger A Y^\dagger
    -
    \frac12
    \left\{
    X^\dagger Y^\dagger,A
    \right\}.
\end{equation}
\end{proposition}

\begin{proof}
For the first contribution,
\begin{align*}
    \left\langle
    A,XBY
    \right\rangle_{\mathrm{HS}}
    &=
    \operatorname{Tr}
    \left(
    A^\dagger XBY
    \right)
    \\
    &=
    \operatorname{Tr}
    \left(
    YA^\dagger XB
    \right)
    \\
    &=
    \operatorname{Tr}
    \left[
    \left(
    X^\dagger A Y^\dagger
    \right)^\dagger
    B
    \right].
\end{align*}
Therefore, the Hilbert-Schmidt adjoint of \(B\mapsto XBY\) is \(A\mapsto X^\dagger A Y^\dagger\). Similarly,
\[
    L_{YX}^\ddagger
    =
    L_{X^\dagger Y^\dagger},
    \qquad
    R_{YX}^\ddagger
    =
    R_{X^\dagger Y^\dagger},
\]
which gives \eqref{eq:HS-adjoint-universal-D}.
\end{proof}

For the physical elementary blocks, this result becomes
\begin{equation}\label{eq:HS-adjoint-physical-D}
    \mathfrak{D}_{\alpha\beta}^\ddagger(A)
    =
    L_\alpha^\dagger A L_\beta
    -
    \frac12
    \left\{
    L_\alpha^\dagger L_\beta,A
    \right\}.
\end{equation}
Evaluating the dual map on the identity yields
\begin{align}
    \mathfrak{D}_{\alpha\beta}^\ddagger(\mathbb I)
    &=
    L_\alpha^\dagger L_\beta
    -
    \frac12
    \left\{
    L_\alpha^\dagger L_\beta,\mathbb I
    \right\}
    \nonumber\\
    &=
    0.
\label{eq:dual-annihilates-identity}
\end{align}
Equation \eqref{eq:dual-annihilates-identity} represents the Heisenberg-picture analogue of the universal trace condition: preserving the trace in the Schrödinger picture is equivalent to the dual dynamics leaving the identity operator unchanged.

\begin{remark}
In general,
\[
    \mathcal D(X,Y)^\ddagger
    \neq
    \mathcal D(Y^\dagger,X^\dagger).
\]
The first expression is the Hilbert-Schmidt adjoint of a superoperator, whereas the second appears in the Hermitian conjugation of its output, as shown in \eqref{eq:universal-Hermitian-conjugation}.
\end{remark}

\section{Qubit realization of the universal dissipative tensor}\label{Sec:QubitRealization}

We next demonstrate the universal construction in the simplest nontrivial operator space. In addition to giving an explicit form for the tensors \(B_{\mu\nu}{}^\lambda\), \(C_{\mu\nu}{}^\lambda\), and \(\mathcal D_{\mu\nu\sigma}{}^\delta\), the qubit case clearly exhibits the trace, Hermiticity, and reality properties derived above.

\subsection{Lie-Jordan structure of the qubit operator space}

We choose the normalized Hermitian basis
\begin{equation}\label{eq:qubit-HS-basis}
    h_0
    =
    \frac{\mathbb I}{\sqrt2},
    \qquad
    h_a
    =
    \frac{\sigma_a}{\sqrt2},
    \qquad
    a=1,2,3,
\end{equation}
where \(\sigma_a\) are the Pauli matrices. This basis satisfies \(\operatorname{Tr}(h_\mu h_\nu) = \delta_{\mu\nu}\), using \(\sigma_a\sigma_b = \delta_{ab}\mathbb I + \mathrm{i}\epsilon_{ab}{}^c\sigma_c\), the ordered products become
\begin{subequations}\label{eq:qubit-ordered-products}
\begin{align}
    h_0h_0
    &=
    \frac1{\sqrt2}h_0,
    \\
    h_0h_a
    =
    h_ah_0
    &=
    \frac1{\sqrt2}h_a,
    \\
    h_ah_b
    &=
    \frac{\delta_{ab}}{\sqrt2}h_0
    +
    \frac{\mathrm{i}}{\sqrt2}
    \epsilon_{ab}{}^c h_c.
\end{align}
\end{subequations}

The nonvanishing Jordan structure constants are therefore
\begin{subequations}\label{eq:qubit-B-constants}
\begin{align}
    B_{00}{}^0
    &=
    \sqrt2,
    \\
    B_{0a}{}^b
    =
    B_{a0}{}^b
    &=
    \sqrt2\,\delta_a{}^b,
    \\
    B_{ab}{}^0
    &=
    \sqrt2\,\delta_{ab},
\end{align}
\end{subequations}
whereas the nonvanishing Lie structure constants are
\begin{equation}\label{eq:qubit-C-constants}
    C_{ab}{}^c
    =
    \mathrm{i}\sqrt2\,
    \epsilon_{ab}{}^c.
\end{equation}

\subsection{Universal elementary dissipative blocks}

\begin{proposition}[Universal qubit blocks]
For the basis \eqref{eq:qubit-HS-basis}, the universal elementary maps satisfy
\begin{subequations}\label{eq:qubit-universal-D-blocks}
\begin{align}
    \mathcal D_{00}(Z)
    &=
    0,
    \\
    \mathcal D_{a0}(Z)
    &=
    \frac1{2\sqrt2}
    [h_a,Z],
    \\
    \mathcal D_{0a}(Z)
    &=
    -\frac1{2\sqrt2}
    [h_a,Z].
\end{align}
\end{subequations}
For \(a,b,c=1,2,3\), their action on the basis elements is
\begin{subequations}\label{eq:qubit-Dab-action}
\begin{align}
    \mathcal D_{ab}(h_0)
    &=
    \mathrm{i}\epsilon_{ab}{}^k h_k,
    \label{eq:qubit-Dab-h0}
    \\
    \mathcal D_{ab}(h_c)
    &=
    \frac12
    \left(
    \delta_{ac}h_b
    +
    \delta_{bc}h_a
    \right)
    -
    \delta_{ab}h_c.
    \label{eq:qubit-Dab-hc}
\end{align}
\end{subequations}
\end{proposition}

\begin{proof}
Since \(h_0=\mathbb I/\sqrt2\),
\[
\mathcal D_{00}(Z)
=
\frac12Z
-
\frac12
\left\{
\frac{\mathbb I}{2},Z
\right\}
=
0.
\]
Likewise,
\begin{align*}
\mathcal D_{a0}(Z)
&=
\frac1{\sqrt2}h_aZ
-
\frac1{2\sqrt2}
\{h_a,Z\}
\\
&=
\frac1{2\sqrt2}
[h_a,Z],
\end{align*}
and the expression for \(\mathcal D_{0a}\) follows analogously.

For the remaining blocks,
\begin{align*}
\mathcal D_{ab}(h_0)
&=
h_ah_0h_b
-
\frac12\{h_bh_a,h_0\}
\\
&=
\frac1{\sqrt2}
(h_ah_b-h_bh_a)
\\
&=
\frac1{\sqrt2}[h_a,h_b]
\\
&=
\mathrm{i}\epsilon_{ab}{}^k h_k.
\end{align*}
Finally, repeated use of
Eq.~\eqref{eq:qubit-ordered-products} gives
\[
\mathcal D_{ab}(h_c)
=
\frac12
\left(
\delta_{ac}h_b
+
\delta_{bc}h_a
\right)
-
\delta_{ab}h_c.
\]
\end{proof}

Equations \eqref{eq:qubit-Dab-h0} and \eqref{eq:qubit-Dab-hc} specify the full universal dissipative tensor for a qubit. In particular, the fact that no \(h_0\) component appears in any output explicitly confirms that \(\mathcal D_{\mu\nu\sigma}{}^0 = 0\). Furthermore,
\(
\left[
\mathcal D_{ab}(h_\sigma)
\right]^\dagger
=
\mathcal D_{ba}(h_\sigma)
\), in agreement with  \eqref{eq:elementary-Hermitian-conjugation}.

\subsection{Pure dephasing}

Consider a single Hermitian Lindblad operator
\begin{equation}\label{eq:qubit-dephasing-L}
    L
    =
    \sqrt{\kappa}\,h_3
    =
    \sqrt{\frac{\kappa}{2}}\,\sigma_3.
\end{equation}
The corresponding dissipator is
\[
    \mathcal L_{\mathrm{deph}}(Z)
    =
    \kappa\mathcal D_{33}(Z).
\]
Using Eq.~\eqref{eq:qubit-Dab-hc}, we obtain
\begin{subequations}\label{eq:dephasing-basis-action}
\begin{align}
    \mathcal L_{\mathrm{deph}}(h_0)
    &=
    0,
    \\
    \mathcal L_{\mathrm{deph}}(h_1)
    &=
    -\kappa h_1,
    \\
    \mathcal L_{\mathrm{deph}}(h_2)
    &=
    -\kappa h_2,
    \\
    \mathcal L_{\mathrm{deph}}(h_3)
    &=
    0.
\end{align}
\end{subequations}
Therefore, in the ordered basis
\((h_0,h_1,h_2,h_3)\),
\begin{equation}\label{eq:dephasing-matrix}
    \left[
    \mathcal L_{\mathrm{deph}}
    \right]
    =
    \begin{pmatrix}
    0 & 0       & 0       & 0\\
    0 & -\kappa & 0       & 0\\
    0 & 0       & -\kappa & 0\\
    0 & 0       & 0       & 0
    \end{pmatrix}.
\end{equation}

Writing the density matrix as
\begin{equation}\label{eq:qubit-Bloch-expansion}
    \rho
    =
    \frac12
    \left(
    \mathbb I+r_a\sigma_a
    \right)
    =
    \frac1{\sqrt2}
    \left(
    h_0+r_ah_a
    \right),
\end{equation}
the Bloch-vector equations are
\begin{equation}\label{eq:dephasing-Bloch-equations}
    \dot r_1
    =
    -\kappa r_1,
    \qquad
    \dot r_2
    =
    -\kappa r_2,
    \qquad
    \dot r_3
    =
    0.
\end{equation}
The population component is unchanged, whereas the off-diagonal coherences decay exponentially. The first row of \eqref{eq:dephasing-matrix} is zero, reflecting trace preservation, and its first column is zero because the channel is unital.

\subsection{Amplitude damping}

We now consider the Lindblad operator
\begin{equation}\label{eq:amplitude-damping-L}
    L
    =
    \sqrt{\gamma}\,
    \lvert0\rangle\langle1\rvert
    =
    \sqrt{\frac{\gamma}{2}}
    \left(
    h_1+\mathrm{i} h_2
    \right),
\end{equation}
where \( \sigma_3 = \lvert0\rangle\langle0\rvert - \lvert1\rangle\langle1\rvert \). The only nonvanishing expansion coefficients are
\[
    \ell^1
    =
    \sqrt{\frac{\gamma}{2}},
    \qquad
    \ell^2
    =
    \mathrm{i}\sqrt{\frac{\gamma}{2}}.
\]
Consequently,
\begin{equation}\label{eq:amplitude-damping-D-decomposition}
    \mathcal L_{\mathrm{AD}}
    =
    \frac{\gamma}{2}
    \left(
    \mathcal D_{11}
    +
    \mathcal D_{22}
    -
    \mathrm{i}\mathcal D_{12}
    +
    \mathrm{i}\mathcal D_{21}
    \right).
\end{equation}

Using equations \eqref{eq:qubit-Dab-h0} and \eqref{eq:qubit-Dab-hc}, one finds \begin{subequations}\label{eq:amplitude-damping-basis-action}
\begin{align}
    \mathcal L_{\mathrm{AD}}(h_0)
    &=
    \gamma h_3,
    \\
    \mathcal L_{\mathrm{AD}}(h_1)
    &=
    -\frac{\gamma}{2}h_1,
    \\
    \mathcal L_{\mathrm{AD}}(h_2)
    &=
    -\frac{\gamma}{2}h_2,
    \\
    \mathcal L_{\mathrm{AD}}(h_3)
    &=
    -\gamma h_3.
\end{align}
\end{subequations}
Thus, in the ordered basis \((h_0,h_1,h_2,h_3)\),
\begin{equation}\label{eq:amplitude-damping-matrix}
    \left[
    \mathcal L_{\mathrm{AD}}
    \right]
    =
    \begin{pmatrix}
    0      & 0             & 0             & 0\\
    0      & -\gamma/2     & 0             & 0\\
    0      & 0             & -\gamma/2     & 0\\
    \gamma & 0             & 0             & -\gamma
    \end{pmatrix}.
\end{equation}
The first row again vanishes, indicating that the trace is preserved. However, unlike the dephasing channel, the first column is now nonzero, with \(\mathcal L_{\mathrm{AD}}(h_0) = \gamma h_3\). Consequently, the amplitude-damping dynamics are not unital.

Using \eqref{eq:qubit-Bloch-expansion}, the corresponding Bloch vector equations are
\begin{equation}\label{eq:amplitude-damping-Bloch-equations}
    \dot r_1
    =
    -\frac{\gamma}{2}r_1,
    \qquad
    \dot r_2
    =
    -\frac{\gamma}{2}r_2,
    \qquad
    \dot r_3
    =
    \gamma(1-r_3).
\end{equation}

The affine shift of the Bloch vector originates from how the linear superoperator acts on the fixed identity component \(\rho^0 = 1/\sqrt2\). Consequently, the seemingly affine dynamics of the three-dimensional Bloch vector are actually just the restriction of a linear evolution defined on the full four-dimensional operator space.

Both channel matrices, \eqref{eq:dephasing-matrix} and \eqref{eq:amplitude-damping-matrix}, are real, consistent with the general statement in \eqref{eq:physical-dissipative-matrix-real}.

\section{Acknowledgements}\label{sec:Acknowledgements}
GLA acknowledges the financial support of SECIHTI through
a Master's scholarship (CVU No.~2051081).
LB acknowledges the financial support of SECIHTI through
a PhD scholarship (CVU No.~960690).
AKB, VGIS, JCSS, and JLC acknowledge the financial support
of DCB UAM-A through grants 22322035 and 22322036.
VACB acknowledges the financial support of DCB UAM-A.
\section{Conclusions}
The main result of this work is that the dissipative sector of
the GKLS equation possesses a universal algebraic structure.
Specifically, we have shown that every Lindblad dissipator can be
obtained from a single universal trilinear map whose components
are determined entirely by the Lie and Jordan structure tensors
of the operator basis. The dependence on the particular physical
system enters only through the coefficients of the Lindblad
operators and the Kossakowski matrix.
The resulting construction shows that the dissipative part
of the GKLS equation is governed by a universal trilinear map,
\[
\mathcal D(X,Y)Z
=
XZY-\frac12\{YX,Z\},
\]
whose operator-space components are determined entirely by the structure tensors \(B_{\mu\nu}{}^\lambda\) and \(C_{\mu\nu}{}^\lambda\). The dependence on a particular physical model enters only through the coefficients defining the Lindblad operators and through the Kossakowski matrix. This provides a clear separation between the universal operator algebra and the parameters specifying a given dissipative process.

A central result is the identification of the combinations \(B+C\) and \(B-C\) with the left and right regular actions of the associative operator algebra. The elementary dissipative map can therefore be understood as a left-right action corrected by an ordered Jordan term. In this language, the commutativity of the left and right transports follows directly from associativity, clarifying the origin of the corresponding superoperator identity found in Ref. \cite{Bixano:2026qce}. The affine interpretation introduced here gives a geometric meaning to these algebraic operations without requiring any additional structure on the space of density operators.

We have also established the main consistency properties of the universal dissipative map. Its image is contained in the traceless operator subspace, which guarantees trace preservation independently of the chosen basis. Its Hermitian-conjugation rule explains why individual off-diagonal dissipative blocks need not preserve Hermiticity, whereas their contraction with a Hermitian Kossakowski matrix does. In a Hermitian Hilbert-Schmidt basis, the complete GKLS equation is consequently represented by a real system of equations for the operator components. The Hilbert-Schmidt adjoint further provides the corresponding dual description and makes explicit the relation between trace preservation in the Schrödinger picture and the invariance of the identity operator in the Heisenberg picture.

The qubit realization demonstrates that the formalism reproduces the standard pure-dephasing and amplitude-damping dynamics while retaining a direct interpretation in terms of the universal tensor. In particular, the examples distinguish naturally between unital and non-unital evolution and show how the affine translation of the Bloch vector arises from a linear transformation on the full operator space.

The present formulation therefore supplies a structural interpretation of the universal superoperator framework developed in \cite{Bixano:2026qce}. It is particularly suited to extensions in which the operator-space symmetries, invariant subspaces, or dimensional reductions of GKLS generators are relevant. Applications to higher-dimensional systems and to symmetry-adapted dissipative models constitute natural directions for further investigation.

\section{Acknowledgements}\label{sec:Acknowledgements}

LB acknowledges the financial support of SECIHTI through
a PhD scholarship (CVU No.~960690).
GLA acknowledges the financial support of SECIHTI through
a Master's scholarship (CVU No.~2051081).
AKB, VGIS, JCSS, and JLC acknowledge the financial support
of DCB UAM-A through grants 22322035 and 22322036.
VACB acknowledges the financial support of DCB UAM-A.

\appendix

\section{Alternative superoperator-block derivation of the first dissipative term}
\label{App:SuperoperatorBlocks}

The compact derivation of eq \eqref{eq:LrhoL-BC-direct-expanded} follows directly from associativity and the ordered products \(B\pm C\). In this appendix, we provide an independent derivation of the same result using the Lie and Jordan superoperators introduced in Sec.~\ref{SubSec:Super-operators}. This calculation makes explicit how the pure Lie, pure Jordan, and mixed Lie-Jordan blocks combine to reproduce the ordered-product expression.

Using
\(
XY
=
\frac12
\left(
\acom{X}{Y}
+
\comm{X}{Y}
\right),
\quad
YX
=
\frac12
\left(
\acom{X}{Y}
-
\comm{X}{Y}
\right),
\)
with \(X=\rho\) and \(Y=L_\beta^\dagger\), we obtain
\[
\begin{aligned}
\rho L_\beta^\dagger
&=
\frac12
\left(
\acom{\rho}{L_\beta^\dagger}
+
\comm{\rho}{L_\beta^\dagger}
\right)
\\
&=
\frac12
\left(
\acom{L_\beta^\dagger}{\rho}
-
\comm{L_\beta^\dagger}{\rho}
\right)
\\
&=
\frac12
\left(
\widetilde e_{L_\beta^\dagger}
-
e_{L_\beta^\dagger}
\right)(\rho).
\end{aligned}
\]
Here,
\[
\begin{aligned}
e_{L_\alpha}
&=
\ell_\alpha{}^\mu e_\mu,
&
\widetilde e_{L_\alpha}
&=
\ell_\alpha{}^\mu\widetilde e_\mu,
\\
e_{L_\beta^\dagger}
&=
\overline{\ell_\beta{}^\nu}\,e_\nu,
&
\widetilde e_{L_\beta^\dagger}
&=
\overline{\ell_\beta{}^\nu}\,\widetilde e_\nu.
\end{aligned}
\]

Likewise, for every \(X\in\mathcal B(\mathcal H)\),
\[
L_\alpha X
=
\frac12
\left(
\widetilde e_{L_\alpha}
+
e_{L_\alpha}
\right)(X).
\]
Consequently,
\begin{equation}\label{eq:App-LrhoL-factorized}
L_\alpha\rho L_\beta^\dagger
=
\frac14
\left(
\widetilde e_{L_\alpha}
+
e_{L_\alpha}
\right)
\left(
\widetilde e_{L_\beta^\dagger}
-
e_{L_\beta^\dagger}
\right)(\rho).
\end{equation}

Given two superoperators \(A\) and \(B\), their composition can be expressed as
\[
AB
=
\frac12
\left(
\acom{A}{B}
+
\comm{A}{B}
\right).
\]
Expanding \eqref{eq:App-LrhoL-factorized}, we therefore obtain
\begin{align}
L_\alpha\rho L_\beta^\dagger
&=
\frac18
\Big(
\acom{\widetilde e_{L_\alpha}}
      {\widetilde e_{L_\beta^\dagger}}
-
\acom{\widetilde e_{L_\alpha}}
      {e_{L_\beta^\dagger}}
+
\acom{e_{L_\alpha}}
      {\widetilde e_{L_\beta^\dagger}}
-
\acom{e_{L_\alpha}}
      {e_{L_\beta^\dagger}}
\notag\\
&\hspace{1.3cm}
+
\comm{\widetilde e_{L_\alpha}}
      {\widetilde e_{L_\beta^\dagger}}
-
\comm{\widetilde e_{L_\alpha}}
      {e_{L_\beta^\dagger}}
+
\comm{e_{L_\alpha}}
      {\widetilde e_{L_\beta^\dagger}}
-
\comm{e_{L_\alpha}}
      {e_{L_\beta^\dagger}}
\Big)(\rho).
\label{eq:LrhoL-super-expanded}
\end{align}
This is the complete Lie-Jordan superoperator decomposition of the first term in \eqref{eq:D-alpha-beta}.

We now evaluate the eight blocks on a basis element \(h_\sigma\).

The pure Jordan anticommutator block is
\[
\begin{aligned}
\acom{\widetilde e_{L_\alpha}}
      {\widetilde e_{L_\beta^\dagger}}(h_\sigma)
&=
\ell_\alpha{}^\mu
\overline{\ell_\beta{}^\nu}
\Big(
B_{\nu\sigma}{}^pB_{\mu p}{}^\delta
+
B_{\mu\sigma}{}^pB_{\nu p}{}^\delta
\Big)h_\delta.
\end{aligned}
\]

The first mixed Jordan-Lie block is
\[
\begin{aligned}
\acom{\widetilde e_{L_\alpha}}
      {e_{L_\beta^\dagger}}(h_\sigma)
&=
\ell_\alpha{}^\mu
\overline{\ell_\beta{}^\nu}
\Big(
C_{\nu\sigma}{}^pB_{\mu p}{}^\delta
+
B_{\mu\sigma}{}^pC_{\nu p}{}^\delta
\Big)h_\delta
\\
&=
\ell_\alpha{}^\mu
\overline{\ell_\beta{}^\nu}
D_{\mu\nu\sigma}{}^\delta h_\delta,
\end{aligned}
\]
where \(D_{\mu\nu\sigma}{}^\delta\) is the mixed tensor defined in \eqref{eq:D-tensor}.

The second mixed Lie-Jordan block is
\[
\begin{aligned}
\acom{e_{L_\alpha}}
      {\widetilde e_{L_\beta^\dagger}}(h_\sigma)
&=
\ell_\alpha{}^\mu
\overline{\ell_\beta{}^\nu}
\Big(
B_{\nu\sigma}{}^pC_{\mu p}{}^\delta
+
C_{\mu\sigma}{}^pB_{\nu p}{}^\delta
\Big)h_\delta
\\
&=
\ell_\alpha{}^\mu
\overline{\ell_\beta{}^\nu}
D_{\nu\mu\sigma}{}^\delta h_\delta.
\end{aligned}
\]

The pure Lie anticommutator block is
\[
\begin{aligned}
\acom{e_{L_\alpha}}
      {e_{L_\beta^\dagger}}(h_\sigma)
&=
\ell_\alpha{}^\mu
\overline{\ell_\beta{}^\nu}
\Big(
C_{\nu\sigma}{}^pC_{\mu p}{}^\delta
+
C_{\mu\sigma}{}^pC_{\nu p}{}^\delta
\Big)h_\delta.
\end{aligned}
\]

Using the closure relations of section \ref{SubSec:Super-operators}, the four commutator blocks are
\[
\begin{aligned}
\comm{\widetilde e_{L_\alpha}}
      {\widetilde e_{L_\beta^\dagger}}(h_\sigma)
&=
\ell_\alpha{}^\mu
\overline{\ell_\beta{}^\nu}
C_{\mu\nu}{}^pC_{p\sigma}{}^\delta h_\delta,
\\
\comm{\widetilde e_{L_\alpha}}
      {e_{L_\beta^\dagger}}(h_\sigma)
&=
\ell_\alpha{}^\mu
\overline{\ell_\beta{}^\nu}
C_{\mu\nu}{}^pB_{p\sigma}{}^\delta h_\delta,
\\
\comm{e_{L_\alpha}}
      {\widetilde e_{L_\beta^\dagger}}(h_\sigma)
&=
\ell_\alpha{}^\mu
\overline{\ell_\beta{}^\nu}
C_{\mu\nu}{}^pB_{p\sigma}{}^\delta h_\delta,
\\
\comm{e_{L_\alpha}}
      {e_{L_\beta^\dagger}}(h_\sigma)
&=
\ell_\alpha{}^\mu
\overline{\ell_\beta{}^\nu}
C_{\mu\nu}{}^pC_{p\sigma}{}^\delta h_\delta.
\end{aligned}
\]

The commutator contributions in  \eqref{eq:LrhoL-super-expanded} cancel pairwise:
\[
C_{\mu\nu}{}^pC_{p\sigma}{}^\delta
-
C_{\mu\nu}{}^pB_{p\sigma}{}^\delta
+
C_{\mu\nu}{}^pB_{p\sigma}{}^\delta
-
C_{\mu\nu}{}^pC_{p\sigma}{}^\delta
=
0.
\]

The remaining contribution from the anticommutator can be expressed as a sum of two ordered products:
\begin{align}
L_\alpha\rho L_\beta^\dagger
&=
\frac18
\ell_\alpha{}^\mu
\overline{\ell_\beta{}^\nu}
\rho^\sigma
\Big[
\left(
B_{\mu\sigma}{}^p
+
C_{\mu\sigma}{}^p
\right)
\left(
B_{\nu p}{}^\delta
-
C_{\nu p}{}^\delta
\right)
\notag\\
&\hspace{3.2cm}
+
\left(
B_{\nu\sigma}{}^p
-
C_{\nu\sigma}{}^p
\right)
\left(
B_{\mu p}{}^\delta
+
C_{\mu p}{}^\delta
\right)
\Big]h_\delta.
\label{eq:App-LrhoL-two-orderings}
\end{align}

Associativity of the operator product implies
\[
(h_\mu h_\sigma)h_\nu
=
h_\mu(h_\sigma h_\nu).
\]
In terms of the Lie--Jordan structure tensors, this is precisely
\[
\left(
B_{\mu\sigma}{}^p
+
C_{\mu\sigma}{}^p
\right)
\left(
B_{\nu p}{}^\delta
-
C_{\nu p}{}^\delta
\right)
=
\left(
B_{\nu\sigma}{}^p
-
C_{\nu\sigma}{}^p
\right)
\left(
B_{\mu p}{}^\delta
+
C_{\mu p}{}^\delta
\right),
\]
in agreement with \eqref{eq:BC-associativity-identity}. Hence the two terms in \eqref{eq:App-LrhoL-two-orderings} coincide, and we obtain
\begin{equation}\label{eq:App-LrhoL-BC-expanded}
L_\alpha\rho L_\beta^\dagger
=
\frac14
\ell_\alpha{}^\mu
\overline{\ell_\beta{}^\nu}
\rho^\sigma
\left(
B_{\mu\sigma}{}^p
+
C_{\mu\sigma}{}^p
\right)
\left(
B_{\nu p}{}^\delta
-
C_{\nu p}{}^\delta
\right)h_\delta.
\end{equation}

Equation~\eqref{eq:App-LrhoL-BC-expanded} coincides with \eqref{eq:LrhoL-BC-direct-expanded}. This establishes the equivalence between the compact ordered-product derivation presented in the main text and the full Lie–Jordan superoperator block decomposition.

The remaining anticommutator term is computed in \eqref{eq:anticommutator-short}. Combining this result with \eqref{eq:App-LrhoL-BC-expanded} yields the universal tensor
\(\mathcal D_{\mu\nu\sigma}{}^\delta\) in \eqref{eq:universal-D-components}, the Lindblad-operator-resolved kernel in \eqref{eq:frakD-final}, and the full component form
\eqref{eq:full-GKLS-components}.

\bibliography{paper-two}
\end{document}